\documentclass[11pt]{article}
\usepackage{amsmath}
\usepackage{amssymb}
\usepackage{amsfonts}

\usepackage{cite}
\usepackage{graphicx}
\usepackage{float}
\usepackage{longtable}
\usepackage[utf8]{inputenc}
\usepackage{hyperref}

\usepackage{algorithmicx}
\usepackage[ruled,vlined]{algorithm2e}
\SetAlgoCaptionSeparator{ }

\usepackage{fullpage}

\newcommand{\fref}[1]{Fig.~\ref{#1}}

\newcommand{\tref}[1]{Table~\ref{#1}}
\newcommand{\sref}[1]{Section~\ref{#1}}

\newenvironment{algo}[1][!htbp]
  {
   \begin{algorithm}[#1]%
  }{\end{algorithm}}

\newenvironment{proced}[1][!h]
  {
   \begin{algorithm}[#1]%
  }{\end{algorithm}}

\newenvironment{subproc}[1][!htbp]
  {
   \begin{algorithm}[#1]%
  }{\end{algorithm}}

\newenvironment{subproc2}[1][!h]
  {
   \begin{algorithm}[#1]%
  }{\end{algorithm}}

\providecommand{\U}[1]{\protect\rule{.1in}{.1in}}
\newtheorem{theorem}{Theorem}

\newtheorem{lemma}{Lemma}

\newenvironment{proof}[1][Proof]{\textbf{#1.} }{\ \rule{0.5em}{0.5em}}
\begin{document}

\title{Quantum Circuit Design for Objective Function Maximization in Gate-Model Quantum Computers}
\author{Laszlo Gyongyosi\thanks{School of Electronics and Computer Science, University of Southampton, Southampton SO17 1BJ, U.K., and Department of Networked Systems and Services, Budapest University of Technology and Economics, 1117 Budapest, Hungary, and MTA-BME Information Systems Research Group, Hungarian Academy of Sciences, 1051 Budapest, Hungary.}
\and Sandor Imre\thanks{Department of Networked Systems and Services, Budapest University of Technology and Economics, 1117 Budapest, Hungary.}}\date{}

\maketitle
\begin{abstract}
Gate-model quantum computers provide an experimentally implementable architecture for near term quantum computations. To design a reduced quantum circuit that can simulate a high complexity reference quantum circuit, an optimization should be taken on the number of input quantum states, on the unitary operations of the quantum circuit, and on the number of output measurement rounds. Besides the optimization of the physical layout of the hardware layer, the quantum computer should also solve difficult computational problems very efficiently. To yield a desired output system, a particular objective function associated with the computational problem fed into the quantum computer should be maximized. The reduced gate structure should be able to produce the maximized value of the objective function. These parallel requirements must be satisfied simultaneously, which makes the optimization difficult. Here, we demonstrate a method for designing quantum circuits for gate-model quantum computers and define the Quantum Triple Annealing Minimization (QTAM) algorithm. The aim of QTAM is to determine an optimal reduced topology for the quantum circuits in the hardware layer at the maximization of the objective function of an arbitrary computational problem.
\end{abstract}

\section{Introduction}
\label{sec1}
According to Moore’s law \cite{ref31}, traditional computer architectures will reach their physical limits in the near future. Quantum computers \cite{ref1, ref2, ref3, ref4, ref5, ref6, ref7,ref8, ref9, ref10, ref11, ref12, ref13, ref14, ref15, ref16, ref17, ref18, ref19, ref20, ref21, ref22,ref23} provide a tool to solve problems more efficiently than ever would be possible with traditional computers \cite{ref1, ref2, ref3, ref4, ref5, ref6, ref7,ref8, ref9, ref10, ref11}. The power of quantum computing is based on the fundamentals of quantum mechanics. In a quantum computer, information is represented by quantum information, and information processing is achieved by quantum gates that realize quantum operations \cite{ref1, ref2, ref3, ref4, ref5, ref6, ref7,ref8, ref9, ref10, ref11,p1,p2,p3}. These quantum operations are performed on the quantum states, which are then outputted and measured in a measurement phase. The measurement process is applied to each quantum state where the quantum information conveyed by the quantum states is converted into classical bits. Quantum computers have been demonstrated in practice \cite{ref1, ref2, ref3, ref4, ref5, ref6, ref8, ref9}, and several implementations are currently in progress \cite{ref1, ref2, ref3, ref4, ref5, ref6, ref7,ref8, ref9, ref10, ref11, ref16, ref17, ref18, ref19}. 

In the physical layer of a gate-model quantum computer, the device contains quantum gates, quantum ports (of quantum gates), and quantum wires for the quantum circuit\footnote{The term ``quantum circuit'', in general, refers to software, not hardware; it is a description or prescription for what quantum operations should be applied when and does not refer to a physically implemented circuit analogous to a printed electronic circuit. In our setting, it refers to the hardware layer.}. In contrast to traditional automated circuit design \cite{ref24, ref25, ref26, ref27, ref28, ref29, ref30}, a quantum system cannot participate in more than one quantum gate simultaneously. As a corollary, the quantum gates of a quantum circuit are applied in several rounds in the physical layer of the quantum circuit \cite{ref1, ref2, ref3, ref4, ref5, ref6, ref7,ref8, ref9, ref10, ref11, ref16, ref17, ref18, ref19}.

The physical layout design and optimization of quantum circuits have different requirements with several open questions and currently represent an active area of study \cite{ref1, ref2, ref3, ref4, ref5, ref6, ref7,ref8, ref9, ref10, ref11, ref16, ref17, ref18, ref19}. Assuming that the goal is to construct a reduced quantum circuit that can simulate the original system, the reduction process should be taken on the number of input quantum states, gate operations of the quantum circuit, and the number of output measurements. Another important question is the maximization of objective function associated with an arbitrary computational problem that is fed into the quantum computer. These parallel requirements must be satisfied simultaneously, which makes the optimization procedure difficult and is an emerging issue in present and future quantum computer developments. 

In the proposed QTAM method, the goal is to determine a topology for the quantum circuits of quantum computer architectures that can solve arbitrary computational problems such that the quantum circuit is minimized in the physical layer, and the objective function of an arbitrary selected computational problem is maximized. The physical layer minimization covers the simultaneous minimization of the quantum circuit area (quantum circuit height and depth of the quantum gate structure, where the depth refers to the number of time steps required for the quantum operations making up the circuit to be run on quantum hardware), the total area of the quantum wires of the quantum circuit, the maximization of the objective function, and the minimization of the required number of input quantum systems and output measurements. An important aim of the physical layout minimization is that the resulting quantum circuit should be identical to a high complexity reference quantum circuit (i.e., the reduced quantum circuit should be able to simulate a nonreduced quantum circuit). 

The minimization of the total quantum wire length in the physical layout is also an objective in QTAM. It serves to improve the processing in the topology of the quantum circuit. However, besides the minimization of the physical layout of the quantum circuit, the quantum computer also has to solve difficult computational problems very efficiently (such as the maximization of an arbitrary combinatorial optimization objective function \cite{ref16, ref17, ref18, ref19}. To achieve this goal in the QTAM method, we also defined an objective function that provides the maximization of objective functions of arbitrary computational problems. The optimization method can be further tuned by specific input constraints on the topology of the quantum circuit (paths in the quantum circuit, organization of quantum gates, required number of rounds of quantum gates, required number of measurement operators, Hamiltonian minimization, entanglement between quantum states, etc.) or other hardware restrictions of quantum computers, such as the well-known \textit{no-cloning theorem} \cite{ref22}. The various restrictions on quantum hardware, such as the number of rounds required to be integrated into the quantum gate structure, or entanglement generation between the quantum states are included in the scheme. These constraints and design attributes can be handled in the scheme through the definition of arbitrary constraints on the topology of the quantum circuit, or by constraints on the computational paths.  

The combinatorial objective function is measured on a computational basis, and an objective function value is determined from the measurement result to quantify the current state of the quantum computer. Quantum computers can be used for combinatorial optimization problems. These procedures aim to use the quantum computer to produce a quantum system that is dominated by computational basis states such that a particular objective function is maximized. 

Recent experimental realizations of quantum computers are qubit architectures \cite{ref1, ref2, ref3, ref4, ref5, ref6, ref7,ref8, ref9, ref10, ref11, ref12, ref13, ref14, ref15, ref16, ref17, ref18, ref19}, and the current quantum hardware approaches focus on qubit systems (i.e., the dimension $d$ of the quantum system is two, $d=2$). However, while the qubit layout is straightforwardly inspirable by ongoing experiments, the method is developed for arbitrary dimensions to make it applicable for future implementations. Motivated by these assumptions, we therefore would avoid the term `qubit' in our scheme to address the quantum states and instead use the generalized term, `quantum states' throughout, which refers to an arbitrary dimensional quantum system. We also illustrate the results through superconducting quantum circuits \cite{ref1, ref2, ref3, ref4, ref5}; however, the framework is general and flexible, allowing a realization for near term gate-model quantum computer implementations.

The novel contributions of this paper are as follows:
\begin{itemize}
\item \textit{We define a method for designing quantum circuits for gate-model quantum computers.} 
\item \textit{We conceive the QTAM algorithm, which provides a quantum circuit minimization on the physical layout (circuit depth and area), quantum wire length minimization, objective function maximization, input size and measurement size minimization for quantum circuits.} 
\item \textit{We define a multilayer structure for quantum computations using the hardware restrictions on the topology of gate-model quantum computers.} 
\end{itemize}

This paper is organized as follows. In \sref{relw} the related works are summarized. \sref{sec2} proposes the system model. In \sref{sec4} the details of the optimization method are discussed, while \sref{sec5} studies the performance of the model. Finally, \sref{sec6} concludes the paper. Supplemental information is inlucded in the Appendix.

\section{Related Works}
\label{relw}
The related works are summarized as follows. 

A strong theoretical background on the logical model of gate-model quantum computers can be found in \cite{ref17,ref16,ref18}. In \cite{ref7}, the model of a gate-model quantum neural network model is defined.

In \cite{refa1}, the authors defined a hierarchical approach to computer-aided design of quantum circuits. The proposed model was designed for the synthesis of permutation class of quantum logic circuits. The method integrates evolutionary and genetic approaches to evolve arbitrary quantum circuit specified by a target unitary matrix. Instead of circuit optimization, the work focuses on circuit synthesis.

In \cite{refa2}, the authors propose a simulation of quantum circuits by low-rank stabilizer decompositions. The work focuses on the problem of simulation of quantum circuits containing a few non-Clifford gates. The framework focuses on the theoretical description of the stabilizer rank. The authors also derived the simulation cost.

A method for the designing of a T-count optimized quantum circuit for integer multiplication with $4n+1$ qubits was defined in \cite{int}. The T-count \cite{tc} measures the number of T-gates, and has a relevance because of the implementation cost of a T gate is high. The aim of the T-count optimization is to reduce the number of T-gates without substantially increasing the number of qubits. The method also applied for quantum circuit designs of integer division \cite{int2}. In the optimization takes into consideration both the T-count and T-depth, since T-depth is also an important performance measure to reduce the implementation costs. Another method for designing of reversible floating point divider units was proposed in \cite{div}.

In \cite{logic}, a methodology for quantum logic gate construction was defined. The main purpose of the scheme was to construct fault-tolerant quantum logic gates with a simple technique. The method is based on the quantum teleportation method \cite{tel}.

A method for the synthesis of depth-optimal quantum circuits was defined in \cite{depth}. The aim of the proposed algorithm is to compute the depth-optimal decompositions of logical operations via an application of the so-called meet-in-the-middle technique. The authors also applied their scheme for the factorizations of some quantum logical operations into elementary gates in the in the Clifford+T set.

A framework to the study the compilation and description of fault-tolerant, high level quantum circuits is proposed in \cite{ft}. The authors defined a method to convert high level quantum circuits consisting of commonly used gates into a form employing all decompositions and ancillary protocols needed for fault-tolerant error correction. The method also represents a useful tool for quantum hardware architectures with topological quantum codes.

The Quantum Approximate Optimization Algorithm (QAOA) optimization algorithm is defined in \cite{ref16}. The QAOA has been defined to evalute approximate solutions for combinatorial optimization problems fed into the quantum computer. 

Relevant attributes of the QAOA algorithm are studied in \cite{refa3}.

In \cite{refa4}, the authors analyzed the performance of the QAOA algorithm on near-term gate-model quantum devices. 

The implementation of QAOA with parallelizable gates is studied in \cite{refa5}.

In \cite{refa6} the performance of QAOA is studied on different problems. The analysis covers the MaxCut combinatorial optimization problem, and the problem of quantum circuit optimizations on a classical computer using automatic differentiation and stochastic gradient descent. The work also revealed that QAOA can exceed the performance of a classical polynomial time algorithm (Goemans-Williamson algorithm \cite{refgw}) with modest circuit depth. The work also concluded that the performance of QAOA with fixed circuit depth is insensitive to problem size.

In \cite{refa7}, the authors studied the problem of ultrafast state preparation via the QAOA with long range interactions. The works provides an application for the QAOA in near-term gate-model quantum devices. As the authors concluded, the QAOA-based approach leads to an extremely efficient state preparation, for example the method allows us to prepare Greene-Horne-Zeilinger (GHZ) states with $\mathcal{O}\left( 1 \right)$ circuit depth. The results were also demonstrated by several other examples. 

Another experimental approach for the implementation of qubit entanglement and parallel logic operations with a superconducting circuit was presented in \cite{song}. In this work, the authors generated entangled GHZ states with up to 10 qubits connecting to a bus resonator in a superconducting circuit. In the proposed implementation, the resonator-mediated qubit-qubit interactions are used to control the entanglement between the qubits and to operate on different pairs in parallel.

A review on the noisy intermediate-scale quantum (NISQ) era can be found in \cite{refpr}. 

The subject of quantum computational supremacy is discussed in \cite{refha, aar}.

For a survey on the attributes of quantum channels, see \cite{ref11}, a survey on quantum computing technology is included in \cite{refsur}.

\section{System Model}
\label{sec2}
The simultaneous physical-layer minimization and the maximization of the objective function are achieved by the Quantum Triple Annealing Minimization (QTAM) algorithm. The QTAM algorithm utilizes the framework of simulated annealing (SA)  \cite{ref24, ref25, ref26, ref27, ref28, ref29, ref30}, which is a stochastic point-to-point search method. 

The procedure of the QTAM algorithm with the objective functions are depicted in \fref{fig1}. The detailed descriptions of the methods and procedures are included in the next sections. 

 \begin{center}
\begin{figure}[h!]
%\vspace{-0.4cm}
\begin{center}
\includegraphics[angle = 0,width=1\linewidth]{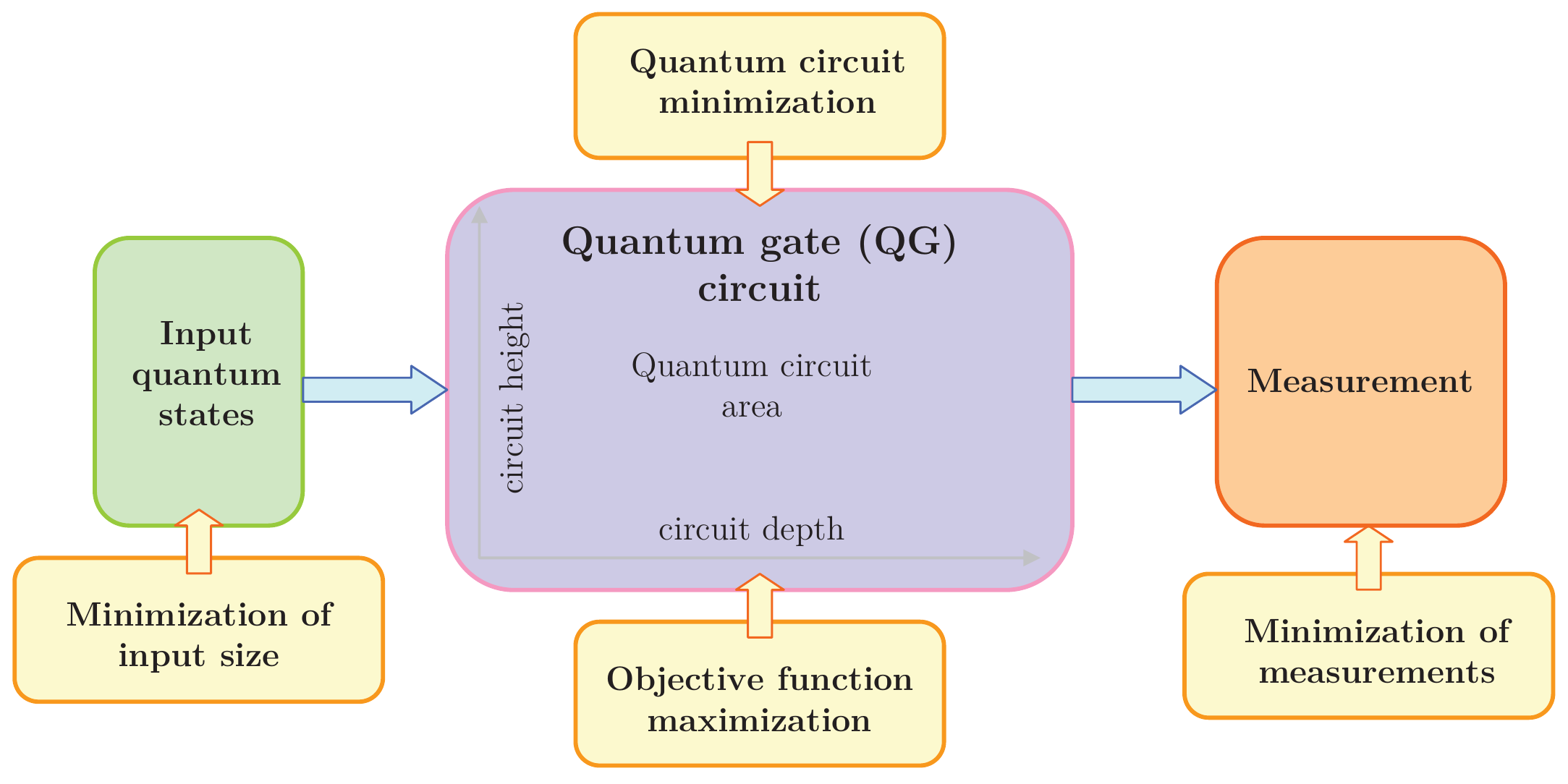}
\caption{The QTAM method for quantum computers. The quantum gate ($QG$) circuit computation model consists of an input array of $n$ quantum states (depicted by the green box), layers of quantum gates integrated into a quantum circuit (depicted by the purple box), and a measurement phase (depicted by the orange box). The quantum gates that act on the quantum states formulate a quantum circuit with a given circuit height and depth. The area of the quantum circuit is minimized by objective function $F_{{\rm 1}} $, while the total quantum wire area of the quantum circuit is minimized by $F_{{\rm 2}} $ ($F_{{\rm 1}} \wedge F_{{\rm 2}} $ is referred via the quantum circuit minimization). The result of the minimization is a quantum circuit of quantum gates with minimized quantum circuit area, minimized total quantum wire length, and a minimized total Hamiltonian operator. The maximization of a corresponding objective function of arbitrary selected computational problems for the quantum computer is achieved by $F_{{\rm 3}} $ (referred via the objective function maximization). Objective functions $F_{{\rm 4}} $ and $F_{{\rm 5}} $ are defined for the minimization of the number of quantum states (minimization of input size), and the total number of measurements (minimization of measurements).} 
 \label{fig1}
 \end{center}
\end{figure}
\end{center}

\subsection{Computational Model}
 By theory, in an SA-based procedure a current solution $s_{A} $ is moved to a neighbor $s_{B} $, which yields an acceptance probability \cite{ref24, ref25, ref26, ref27, ref28, ref29, ref30} 
\begin{equation} \label{eq1} 
{\Pr }\left(f\left(s_{A} \right),f\left(s_{B} \right)\right)=\frac{{\rm 1}}{{\rm 1}+e^{\left(\frac{f\left(s_{A} \right)-f\left(s_{B} \right)}{Tf\left(s_{A} \right)} \right)} } , 
\end{equation} 
 where $f\left(s_{A} \right)$ and $f\left(s_{B} \right)$ represent the relative performances of the current and neighbor solutions, while $T$ is a control parameter, $T\left(t\right)=T_{\max } {\rm exp}\left(-R\left(t/k\right)\right)$, where $R$ is the temperature decreasing rate, $t$ is the iteration counter, $k$ is a scaling factor, while $T_{\max } $ is an initial temperature.

Since SA is a probabilistic procedure it is important to minimize the acceptance probability of unfavorable solutions and avoid getting stuck in a local minima.

Without loss of generality, if $T$ is low, \eqref{eq1} can be rewritten in function of $f\left(s_{A} \right)$ and $f\left(s_{B} \right)$ as 
\begin{equation} \label{eq2} 
{\Pr }\left(f\left(s_{A} \right),f\left(s_{B} \right)\right)=\left\{\begin{split} {{\rm 1,if\text{ }}f\left(s_{A} \right)>f\left(s_{B} \right)} \\ {{\rm 0,if\text{ }}f\left(s_{A} \right)\le f\left(s_{B} \right)} \end{split}\right. . 
\end{equation} 
 In the QTAM algorithm, we take into consideration that the objectives, constraints, and other functions of the method, by some fundamental theory, are characterized by different magnitude ranges \cite{ref24, ref25, ref26, ref27, ref28, ref29, ref30}. To avoid issues from these differences in the QTAM algorithm we define three annealing temperatures, $T_{f} \left(t\right)$ for objectives, $T_{g} \left(t\right)$ for constraints and $T_{c} \left(t\right)$ for the probability distribution closeness (distance of the output distributions of the reference quantum circuit and the reduced quantum circuit).

In the QTAM algorithm, the acceptance probability of a new solution $s_{B} $ at a current solution $s_{A} $ is as 
\begin{equation} \label{eq3} 
{\Pr }\left(s_{A} ,s_{B} \right)=\frac{{\rm 1}}{{\rm 1}+e^{\tilde{d}\left(f\right)T_{f} \left(t\right)} e^{\tilde{d}\left(g\right)T_{g} \left(t\right)} e^{\tilde{d}\left(c\right)T_{c} \left(t\right)} } , 
\end{equation} 
 where $\tilde{d}\left(f\right)$, $\tilde{d}\left(g\right)$ and $\tilde{d}\left(c\right)$ are the average values of objective, constraint and distribution closeness domination, see Algorithm 1.

To aim of the QTAM algorithm is to minimize the cost function 
\begin{equation} \label{eq4} 
\min f\left({\rm x}\right)=\alpha _{{\rm 1}} F_{{\rm 1}} \left({\rm x}\right)+\ldots +\alpha _{N_{obj} } F_{N_{obj} } \left({\rm x}\right)+F_{s} , 
\end{equation} 
 where ${\rm x}$ is the vector of design variables, while $\alpha $ is the vector of weights, while $N_{obj} $ is the number of primarily objectives. Other $i$ secondary objectives (aspect ratio of the quantum circuit, overlaps, total net length, etc.) are minimized simultaneously via the single-objective function $F_{s} $ in \eqref{eq4} as 
\begin{equation} \label{eq5} 
F_{s} =\sum _{i} \alpha _{i} F_{i} \left(x\right). 
\end{equation} 

\subsection{Objective Functions}
We defined $N_{obj} =5$ objective functions for the QTAM algorithm. Objective functions $F_{{\rm 1}} $ and $F_{{\rm 2}} $ are defined for minimization of $QG$ quantum circuit in the physical layer. The aim of objective function $F_{{\rm 1}} $ is the minimization of the $A_{QG} $ quantum circuit area of the $QG$ quantum gate structure, 
\begin{equation} \label{eq6} 
F_{{\rm 1}} :\min \left(A_{QG} \right)=\min \left(H'_{QG} \cdot D'_{QG} \right), 
\end{equation} 
 where $H'_{QG} $ is the optimal circuit height of $QG$, while $D'_{QG} $ is the optimal depth of $QG$.

Focusing on superconducting quantum circuits \cite{ref1, ref2, ref3, ref4, ref5}, the aim of $F_{{\rm 2}} $ is the physical layout minimization of the $w_{QG} $ total quantum wire area of $QG$, as 
\begin{equation} \label{eq7} 
F_{{\rm 2}} :w_{QG} =\min \sum _{k=1}^{h} \left(\sum _{i=1}^{p} \sum _{j=1}^{q} \ell _{ij} \cdot \delta _{ij} \left(\psi _{ij} \right)\right), 
\end{equation} 
 where $h$ is the number of nets of the $QG$ circuit, $p$ is the number of quantum ports of the $QG$ quantum circuit considered as sources of a condensate wave function amplitude \cite{ref1, ref2, ref3, ref4, ref5}, and $q$ the number of quantum ports considered as sinks of a condensate wave function amplitude, $\ell _{ij} $ is the length of the quantum wire $ij$, $\delta _{ij} $ is the effective width of the quantum wire $ij$, while $\psi _{ij} $ is the (root mean square) condensate wave function amplitude \cite{ref1, ref2, ref3, ref4, ref5} associated to the quantum wire $ij$. 

Objective function $F_{{\rm 3}} $ is defined for the maximization of the expected value of an objective function $C_{L} (\vec{\Phi })$ as  
\begin{equation} \label{eq8} 
F_{3} :\max {C_{L} (\vec{\Phi })}=\max {\langle \vec{\Phi }|C|\vec{\Phi }\rangle }, 
\end{equation} 
 where $C$ is an objective function, $\vec{\Phi }$ is a collection of $L$ parameters 
\begin{equation} \label{eq9} 
\vec{\Phi }=\Phi _{{\rm 1}} ,\ldots ,\Phi _{L}  
\end{equation} 
 such that with $L$ unitary operations, state $|\vec{\Phi }\rangle $ is evaluated as 
\begin{equation} \label{eq10} 
|\vec{\Phi }\rangle =U_{L} \left(\Phi _{L} \right),\ldots ,U_{{\rm 1}} \left(\Phi _{{\rm 1}} \right)\left| \varphi \right\rangle  , 
\end{equation} 
 where $U_{i} $ is an $i$-th unitary that depends on a set of parameters $\Phi _{i} $, while $\left| \varphi \right\rangle  $ is an initial state. Thus the goal of $F_{{\rm 3}} $ is to determine the $L$ parameters of $\vec{\Phi }$ (see \eqref{eq9}) such that $\langle \vec{\Phi }|C|\vec{\Phi }\rangle $ is maximized.

Objective functions $F_{{\rm 4}} $ and $F_{{\rm 5}} $ are defined for the minimization of the number of input quantum states and the number of required measurements. The aim of objective function $F_{{\rm 4}} $ is the minimization of the number of quantum systems on the input of the $QG$ circuit, 
\begin{equation} \label{eq11} 
F_{{\rm 4}} :\min \left(n\right). 
\end{equation} 
 The aim of objective function $F_{{\rm 5}} $ is the minimization of the total number of measurements in the $M$ measurement block, 
\begin{equation} \label{eq12} 
F_{{\rm 5}} :\min \left(m\right)=\min {\left(N_{M} \left|M\right|\right)}, 
\end{equation} 
 where $m=N_{M} \left|M\right|$, where $N_{M} $ is the number of measurement rounds, $\left|M\right|$ is the number of measurement gates in the $M$ measurement block. 

\subsection{Constraint Violations}
 The optimization at several different objective functions results in different Pareto fronts \cite{ref24, ref25, ref26, ref27} of placements of quantum gates in the physical layout. These Pareto fronts allow us to find feasible tradeoffs between the optimization objectives of the QTAM method. The optimization process includes diverse objective functions, constraints, and optimization criteria to improve the performance of the quantum circuit and to take into consideration the hardware restrictions of quantum computers. In the proposed QTAM algorithm the constraints are endorsed by the modification of the Pareto dominance \cite{ref24, ref25, ref26, ref27} values by the different sums of constraint violation values. We defined three different constraint violation values.
 
\subsubsection{Distribution Closeness Dominance}
 In the QTAM algorithm, the Pareto dominance is first modified with the sum of distribution closeness violation values, denoted by $c_{s} \left(\cdot \right)$. The aim of this iteration is to support the closeness of output distributions of the reduced quantum circuit $QG$ to the output distribution of the reference quantum circuit $QG_{R} $.

Let $P_{QG_{R} } $ the output distribution after the $M$ measurement phase of the reference (original) quantum circuit $QG_{R} $ to be simulated by $QG$, and let $Q_{QG} $ be the output distribution of the actual, reduced quantum circuit $QG$. The distance between the quantum circuit output distributions $P_{QG_{R} } $ and $Q_{QG} $ (distribution closeness) is straightforwardly yielded by the relative entropy function, as 
\begin{equation} \label{eq13} 
D\left(\left. P_{QG_{R} } \right\| Q_{QG} \right)=\sum _{i} P_{QG_{R} } \left(i\right)\log _{2} \frac{P_{QG_{R} } \left(i\right)}{Q_{QG} \left(i\right)} . 
\end{equation} 
 For two solutions $x$ and $y$, the $d_{x,y} \left(c\right)$ distribution closeness dominance function is defined as 
\begin{equation} \label{eq14} 
d_{x,y} \left(c\right)=c_{s} \left(x\right)-c_{s} \left(y\right), 
\end{equation} 
 where $c_{s} \left(\cdot \right)$ is evaluated for a given solution $z$ as 
\begin{equation} \label{eq15} 
c_{s} \left(z\right)=\sum _{i=1}^{N_{v} } v_{i}^{c} , 
\end{equation} 
 where $v_{i}^{c} $ is an $i$-th distribution closeness violation value, $N_{v} $ is the number of distribution closeness violation values for a solution $z$.

In terms of distribution closeness dominance, $x$ dominates $y$ if the following relation holds: 
\begin{equation} \label{eq16} 
\begin{split} {\left(\left(c_{s} \left(x\right)<0\right)\wedge \left(c_{s} \left(y\right)<0\right)\wedge \left(c_{s} \left(x\right)>c_{s} \left(y\right)\right)\right)} \\ \vee{\left(\left(c_{s} \left(x\right)=0\right)\wedge \left(c_{s} \left(y\right)<0\right)\right),} \end{split} 
\end{equation} 
 thus \eqref{eq16} states that $x$ dominates $y$ if both $x$ and $y$ are unfeasible, and $x$ is closer to feasibility than $y$, or $x$ is feasible and $y$ is unfeasible.

By similar assumptions, $y$ dominates $x$ if 
\begin{equation} \label{eq17} 
\begin{split} {\left(\left(c_{s} \left(x\right)<0\right)\wedge \left(c_{s} \left(y\right)<0\right)\wedge \left(c_{s} \left(x\right)<c_{s} \left(y\right)\right)\right)} \\ \vee{\left(\left(c_{s} \left(x\right)<0\right)\wedge \left(c_{s} \left(y\right)=0\right)\right).} \end{split} 
\end{equation} 
  
\subsubsection{Constraint Dominance}
 The second modification of the Pareto dominance is by the sum of constraint violation values, 
\begin{equation} \label{eq18} 
d_{x,y} \left(g\right)=g_{s} \left(x\right)-g_{s} \left(y\right), 
\end{equation} 
 where $g_{s} \left(\cdot \right)$ is the sum of all constraint violation values, evaluated for a given solution $z$ as 
\begin{equation} \label{eq19} 
g_{s} \left(z\right)=\sum _{i=1}^{N_{g} } v_{i}^{g} , 
\end{equation} 
 where $v_{i}^{g} $ is an $i$-th constraint violation value, $N_{g} $ is the number of constraint violation values for a solution $z$.

Similar to \eqref{eq16} and \eqref{eq17}, in terms of constraint dominance, $x$ dominates $y$ if the following relation holds: 
\begin{equation} \label{eq20} 
\begin{split} {\left(\left(g_{s} \left(x\right)<0\right)\wedge \left(g_{s} \left(y\right)<0\right)\wedge \left(g_{s} \left(x\right)>g_{s} \left(y\right)\right)\right)} \\ \vee{\left(\left(g_{s} \left(x\right)=0\right)\wedge \left(g_{s} \left(y\right)<0\right)\right),} \end{split} 
\end{equation} 
 thus \eqref{eq16} states that $x$ dominates $y$ if both $x$ and $y$ are unfeasible, and $x$ is closer to feasibility than $y$, or $x$ is feasible and $y$ is unfeasible.

By similar assumptions, $y$ dominates $x$ with respect to $g_{s} \left(\cdot \right)$ if 
\begin{equation} \label{eq21} 
\begin{split} {\left(\left(g_{s} \left(x\right)<0\right)\wedge \left(g_{s} \left(y\right)<0\right)\wedge \left(g_{s} \left(x\right)<g_{s} \left(y\right)\right)\right)} \\ \vee{\left(\left(g_{s} \left(x\right)<0\right)\wedge \left(g_{s} \left(y\right)=0\right)\right).} \end{split} 
\end{equation} 
 
\subsubsection{Objective Dominance}
Let $x$ and $y$ refer to two solutions, then, by theory, the $d_{x,y} \left(f\right)$ objective dominance function is defined as 
\begin{equation} \label{eq22} 
d_{x,y} \left(f\right)=\prod _{i=1,f_{{\rm 1}} \left(x\right)\ne f_{{\rm 1}} \left(y\right)}^{N_{obj} } \frac{\left|f_{i} \left(x\right)-f_{i} \left(y\right)\right|}{R_{i} } , 
\end{equation} 
 where $N_{obj} $ is the number of objectives (in our setting $N_{obj} =5$), $R_{i} $ is the range of objective $i$, while $x$ dominates $y$ if $f_{i} \left(x\right)\le f_{i} \left(y\right)$ for $\forall _{i} =1,\ldots ,N_{obj} $, and for at least one $i$ the relation $f_{i} \left(x\right)<f_{i} \left(y\right)$ holds. 
 
\subsection{Objective Function Maximization}
\label{A1}
The quantum circuit $QG$ executes operations in the ${\rm {\mathcal{H}}}$ Hilbert space. The dimension of the ${\rm {\mathcal{H}}}$ space is 
\begin{equation} \label{eq64} 
{\rm dim}\left({\rm {\mathcal{H}}}\right)=d^{n} , 
\end{equation} 
 where $d$ is the dimension of the quantum system ($d=2$ for a qubit system), while $n$ is the number of quantum states.

Using the formalism of \cite{ref16, ref17, ref18}, let assume that the computational problem fed into the quantum circuit $QG$ is specified by $n$ bits and $m$ constraints. Then, the objective function is defined as 
\begin{equation} \label{eq65} 
C\left(z\right)=\sum _{\alpha =1}^{m} C_{\alpha } \left(z\right), 
\end{equation} 
 where 
\begin{equation} \label{eq66} 
z=z_{{\rm 1}} \ldots z_{n}  
\end{equation} 
 is an $n$-length bitstring, and $C_{\alpha } \left(z\right)=1$ if $z$ satisfies constraint $\alpha $, and $C_{\alpha } \left(z\right)=0$ otherwise \cite{ref16, ref17, ref18}.

Assuming a Hilbert space of $n$ qubits, ${\rm dim}\left({\rm {\mathcal{H}}}\right)={\rm 2}^{n} $, using the computational basis vectors $\left| z\right\rangle  $, operator $C\left(z\right)$ in \eqref{eq65} is a diagonal operator in the computational basis \cite{ref16, ref17, ref18}. Then, at a particular angle $\gamma $, $\gamma \in \left[0, \pi \right]$, unitary $U\left(C,\gamma \right)$ is evaluated as 
\begin{equation} \label{eq67} 
U\left(C,\gamma \right)=e^{-i\gamma C} =\prod _{\alpha =1}^{m} e^{-i\gamma C_{\alpha } } , 
\end{equation} 
 such that all terms in the product are diagonal in the computational basis.

Then, for the $\mu $ dependent product of commuting operators, $\mu \in \left[0,\pi \right]$ \cite{ref16, ref17, ref18}, a unitary $U\left(B,\mu \right)$ is defined as 
\begin{equation} \label{eq68} 
U\left(B,\mu \right)=e^{-i\mu B} =\prod _{j=1}^{n} e^{-i\mu \sigma _{x}^{j} } , 
\end{equation} 
 where $B=\sum _{i} X_{i} $, $X_{i} =\sigma _{x}^{i} $, $\sigma _{x} $ is the Pauli $X$-operator, while $\mu $ is a control parameter \cite{ref16, ref17, ref18, ref19}, $\mu \in \left[0,\pi \right]$ . For a qubit setting, the $\left| s\right\rangle  $ initial state of the quantum computer is the uniform superposition over computational basis states, 
\begin{equation} \label{eq69} 
\left| s\right\rangle  =\frac{{\rm 1}}{\sqrt{{\rm 2}^{n} } } \sum _{z} \left| z\right\rangle  . 
\end{equation} 
 Let assume that the $G_{QG}^{k,r} $ multilayer structure of the $QG$ quantum circuit contains $n$ quantum ports of several quantum gates, and edge set 
\begin{equation} \label{eq70} 
{\rm {\mathcal{S}}}_{E} =\left\{\left\langle jk\right\rangle \right\} 
\end{equation} 
 of size $m$. Then, the aim of the optimization is to indentify a string $z$ \eqref{eq66} that the maximizes the objective function 
\begin{equation} \label{eq71} 
C=\sum _{\left\langle jk\right\rangle } C_{\left\langle jk\right\rangle } , 
\end{equation} 
 where 
\begin{equation} \label{eq72} 
C_{\left\langle jk\right\rangle } =\frac{{\rm 1}}{{\rm 2}} \left({\rm 1}-z_{i} z_{j} \right), 
\end{equation} 
 where $z_{i} =\pm 1$.

In $G_{QG}^{k,r} $ different unitary operations can be defined for the single quantum ports (qubits) and the connected quantum ports, as follows.

Let $U_{q_{s} } \left(\mu _{j} \right)$ be a unitary operator on a $q_{s} $ single port (qubits) in $G_{QG}^{k,r} $, be defined such that for each quantum ports a $\mu _{j} $ parameter is associated as 
\begin{equation} \label{eq73} 
U_{q_{s} } \left(\mu _{j} \right)=e^{-i\mu _{j} X_{j} } . 
\end{equation} 
 For the collection 
\begin{equation} \label{eq74} 
\vec{\mu }=\left(\mu _{{\rm 1}} ,\ldots ,\mu _{n} \right), 
\end{equation} 
 the resulting unitary is 
\begin{equation} \label{eq75} 
U_{q_{s} } \left(\vec{\mu }\right)=\prod _{j} U_{q_{s} } \left(\mu _{j} \right). 
\end{equation}

The unitary $U_{q_{s} } \left(\vec{\mu }\right)$ is therefore represents the applications of the unitary operations at once in the quantum ports of the $QG$ quantum circuit.

Then, let unitary $U_{q_{jk} } \left(\gamma _{jk} \right)$ be defined for connected quantum ports $q_{jk} $ in $G_{QG}^{k,r} $, as 
\begin{equation} \label{eq76} 
U_{q_{jk} } \left(\gamma _{jk} \right)=e^{i\gamma _{jk} Z_{j} Z_{k} } , 
\end{equation} 
 where $Z_{i} =\sigma _{z}^{i} $, where $\sigma _{z} $ is the Pauli $Z$-operator. Since the eigenvalues of $X_{i} $ and $Z_{j} Z_{k} $ are $\pm 1$, it allows us to restrict the values \cite{ref16, ref17, ref18} of parameters $\gamma $ and $\mu $ to the range of $\left[0,\pi \right]$.

Then, defining collection 
\begin{equation} \label{eq77} 
\vec{\gamma }=(\gamma _{jk}^{{\rm 1}} ,\ldots ,\gamma _{jk}^{h} ), 
\end{equation} 
 where $h$ is the number of individual $\gamma _{jk} $ parameters, the unitary $U_{q_{jk} } \left(\vec{\gamma }\right)$ is yielded as 
\begin{equation} \label{eq78} 
U_{q_{jk} } \left(\vec{\gamma }\right)=\prod _{\left\langle jk\right\rangle \in G_{QG}^{k,r} } U_{q_{jk} } \left(\gamma _{jk} \right). 
\end{equation} 
 Assuming that there exists a set ${\rm {\mathcal{S}}}_{\vec{\mu }}^{u} $ of $u$ collections of $\vec{\mu }$'s 
\begin{equation} \label{eq79} 
{\rm {\mathcal{S}}}_{\vec{\mu }}^{u} :\vec{\mu }^{\left({\rm 1}\right)} ,\ldots ,\vec{\mu }^{\left(u\right)}  
\end{equation} 
 and a set ${\rm {\mathcal{S}}}_{\vec{\gamma }}^{u} $ of $u$ collections of $\vec{\gamma }$'s, 
\begin{equation} \label{eq80} 
{\rm {\mathcal{S}}}_{\vec{\gamma }}^{u} :\vec{\gamma }^{\left({\rm 1}\right)} ,\ldots ,\vec{\gamma }^{\left(u\right)} , 
\end{equation} 
 a $\left| \phi \right\rangle  $ system state of the $QG$ quantum circuit is evaluated as 
\begin{equation} \label{eq81} 
\begin{split}
 \left| \phi  \right\rangle  &=\left| \mathcal{S}_{{\vec{\mu }}}^{u},\mathcal{S}_{{\vec{\gamma }}}^{u},C \right\rangle  \\ 
 & ={{U}_{{{q}_{s}}}}\left( {{{\vec{\mu }}}^{\left( u \right)}} \right){{U}_{{{q}_{jk}}}}\left( {{{\vec{\gamma }}}^{\left( u \right)}} \right)\ldots {{U}_{{{q}_{s}}}}\left( {{{\vec{\mu }}}^{\left( 1 \right)}} \right){{U}_{{{q}_{jk}}}}\left( {{{\vec{\gamma }}}^{\left( 1 \right)}} \right)\left| s \right\rangle ,  
\end{split}
\end{equation} 
 where $\left| s\right\rangle  $ is given in \eqref{eq69}.

The maximization of objective function \eqref{eq65} in the multilayer $G_{QG}^{k,r} $ structure is therefore analogous to the problem of finding the parameters of sets ${\rm {\mathcal{S}}}_{\vec{\mu }}^{u} $ \eqref{eq79} and ${\rm {\mathcal{S}}}_{\vec{\gamma }}^{u} $ \eqref{eq80} in the system state $\left| \phi \right\rangle  $ \eqref{eq81} of the $QG$ quantum circuit.

\subsection{The QTAM Algorithm}
\label{sec3}
\begin{theorem} The QTAM algorithm utilizes annealing temperatures $T_{f} \left(t\right)$, $T_{g} \left(t\right)$ and $T_{c} \left(t\right)$ to evaluate the acceptance probabilities, where $T_{f} \left(t\right)$ is the annealing temperature for the objectives, $T_{g} \left(t\right)$ is the annealing temperature for the constraints and $T_{c} \left(t\right)$ is the annealing temperature for the distribution closeness.
\end{theorem}
\begin{proof}
The detailed description of the QTAM procedure is given in Algorithm 1. 

\setcounter{algocf}{0}
\begin{algo}
  \DontPrintSemicolon
\caption{\textit{Quantum Triple Annealing Minimization (QTAM)}}
\textbf{Step 1}. Define an archive ${\rm {\mathcal{A}}}$ with random solutions, and select a $\xi $ random solution from ${\rm {\mathcal{A}}}$.

\textbf{Step 2}. Define $\nu $ as $\nu =\Xi \left(\xi \right)$, where $\Xi \left(\cdot \right)$ is a moving operator. Determine the dominance relation between $\xi $ and $\nu $ via ${\rm {\mathcal{D}}}_{P} \left(\xi ,\nu \right)$, where function ${\rm {\mathcal{D}}}_{P} \left(\cdot \right)$ is the constrained Pareto dominance checking function.

\textbf{Step 3}. Evaluate acceptance probabilities based on ${\rm {\mathcal{D}}}_{P} \left(\xi ,\nu \right)$.

\textbf{(a)}: If ${\rm {\mathcal{D}}}_{P} \left(\xi ,\nu \right)=\nu \angle \xi $ ($\xi $ dominates $\nu $, where $\angle $ is the Pareto dominance operator), then $\xi =\nu $, with probability 
\begin{equation} \label{eq23} 
{\Pr }\left(\left. \xi =\nu \right|\nu \angle \xi \right)=\frac{{\rm 1}}{{\rm 1}+e^{\tilde{d}\left(f\right)T_{f} \left(t\right)} e^{\tilde{d}\left(g\right)T_{g} \left(t\right)} e^{\tilde{d}\left(c\right)T_{c} \left(t\right)} } , 
\end{equation} 
 where $\tilde{d}\left(f\right)$, is the average objective dominance, evaluated as 
\begin{equation} \label{eq24} 
\tilde{d}\left(f\right)=\frac{\left(\sum _{i=1}^{k} d_{i,\nu } \left(f\right)\right)+d_{\xi ,\nu } \left(f\right)}{k+{\rm 1}} , 
\end{equation} 
 where $d_{x,y} \left(f\right)$ is the objective dominance function as given in \eqref{eq22}, while $\tilde{d}\left(g\right)$ average constraint dominance as 
\begin{equation} \label{eq25} 
\tilde{d}\left(g\right)=\frac{-\left(\sum _{i=1}^{k} d_{\nu ,i} \left(g\right)\right)-d_{\nu ,\xi } \left(g\right)}{k+{\rm 1}} , 
\end{equation} 
 where $d_{x,y} \left(g\right)$ is the constraint dominance function as given in \eqref{eq18}, and $\tilde{d}\left(c\right)$ is average distribution closeness dominance as 
\begin{equation} \label{eq26} 
\tilde{d}\left(c\right)=\frac{-\left(\sum _{i=1}^{k} d_{\nu ,i} \left(c\right)\right)-d_{\nu ,\xi } \left(c\right)}{k+{\rm 1}} , 
\end{equation} 
 where $d_{x,y} \left(c\right)$ is the distribution closeness dominance function as given in \eqref{eq14}, while $T_{f} \left(t\right)$ is the annealing temperature for the objectives 
\begin{equation} \label{eq27} 
T_{f} \left(t\right)=T_{f_{\max } } e^{-R\left(\frac{t}{k} \right)} , 
\end{equation} 
 where $R$ is the temperature decreasing rate, $T_{f_{\max } } $ is a maximum (initial) value for annealing the objectives factor, $T_{g} \left(t\right)$ is the annealing temperature for the constraints 
\begin{equation} \label{eq28} 
T_{g} \left(t\right)=T_{g_{\max } } e^{-R\left(\frac{t}{k} \right)} , 
\end{equation} 
 where $T_{g_{\max } } $ is a maximum (initial) value for annealing the constraint factor, and $T_{c} \left(t\right)$ is the annealing temperature for the distribution closeness 
\begin{equation} \label{eq29} 
T_{c} \left(t\right)=T_{c_{\max } } e^{-R\left(\frac{t}{k} \right)} , 
\end{equation} 
 where $T_{c_{\max } } $ is a maximum (initial) value for annealing the distribution closeness factor, respectively.  
\end{algo}

\setcounter{algocf}{0}
\begin{algo}
  \DontPrintSemicolon
\caption{\textit{Quantum Triple Annealing Minimization (QTAM), cont.}}
\textbf{(b)}. If ${\rm {\mathcal{D}}}_{P} \left(\xi ,\nu \right)=\left(\nu \neg \angle \xi \right)\wedge \left(\xi \neg \angle \nu \right)$ ($\nu $ and $\xi $ are non-dominating to each other) such that $\nu $ is dominated by $k\ge {\rm 1}$ points in ${\rm {\mathcal{A}}}$, ${\rm {\mathcal{D}}}_{P} \left(\xi ,\nu \right)=\nu \angle \left({\rm {\mathcal{A}}}\right)_{k} $, then $\xi =\nu $, with probability 
\begin{equation} \label{eq30} 
\begin{split} {{\Pr}\left(\left. \xi =\nu \right|\left(\nu \neg \angle \xi \right)\wedge \left(\xi \neg \angle \nu \right),\nu \angle \left({\rm {\mathcal{A}}}\right)_{k} \right)} \\ {=\frac{{\rm 1}}{{\rm 1}+e^{\tilde{d}\left(f\right)_{k} T_{f} \left(t\right)} e^{\tilde{d}\left(g\right)_{k} T_{g} \left(t\right)} e^{\tilde{d}\left(c\right)_{k} T_{c} \left(t\right)} } ,} \end{split} 
\end{equation} 
 where 
\begin{equation} \label{eq31} 
\tilde{d}\left(f\right)_{k} =\tilde{d}\left(f\right)-d_{\xi ,\nu } \left(f\right), 
\end{equation} 
 where $\tilde{d}\left(f\right)$ is as in \eqref{eq24}, 
\begin{equation} \label{eq32} 
\tilde{d}\left(g\right)_{k} =\tilde{d}\left(g\right)+d_{\nu ,\xi } \left(g\right), 
\end{equation} 
 where $\tilde{d}\left(g\right)$ is as in \eqref{eq25}, while 
\begin{equation} \label{eq33} 
\tilde{d}\left(c\right)_{k} =\tilde{d}\left(c\right)+d_{\nu ,\xi } \left(c\right), 
\end{equation} 
 where $\tilde{d}\left(c\right)$ is as in \eqref{eq26}.

\textbf{(c)}: If ${\rm {\mathcal{D}}}_{P} \left(\xi ,\nu \right)=\left(\nu \neg \angle \xi \right)\wedge \left(\xi \neg \angle \nu \right)$, and ${\rm {\mathcal{D}}}_{P} \left(\nu ,{\rm {\mathcal{A}}}\right)=\left({\rm {\mathcal{A}}}\neg \angle \nu \right)$, thus $\nu $ is non-dominating with respect to ${\rm {\mathcal{A}}}$, then apply Sub-procedure 1.

\textbf{(d)}: If ${\rm {\mathcal{D}}}_{P} \left(\xi ,\nu \right)=\left(\nu \neg \angle \xi \right)\wedge \left(\xi \neg \angle \nu \right)$, and ${\rm {\mathcal{D}}}_{P} \left(\nu ,{\rm {\mathcal{A}}}\right)=\left(\left({\rm {\mathcal{A}}}\right)_{k} \angle \nu \right)$, thus $\nu $ dominates $k\ge {\rm 1}$ points in ${\rm {\mathcal{A}}}$, then apply Sub-procedure 2.

\textbf{(e)}: If ${\rm {\mathcal{D}}}_{P} \left(\xi ,\nu \right)=\xi \angle \nu $ such that ${\rm {\mathcal{D}}}_{P} \left(\nu ,{\rm {\mathcal{A}}}\right)=\left(\nu \angle \left({\rm {\mathcal{A}}}\right)_{k} \right)$, thus $\nu $ is dominated by $k\ge {\rm 1}$ points in ${\rm {\mathcal{A}}}$, then set $\xi =\nu $, with probability 
\begin{equation} \label{eq34} 
{\Pr }\left(\left. \xi =\nu \right|\xi \angle \nu ,\nu \angle \left({\rm {\mathcal{A}}}\right)_{k} \right)=\frac{{\rm 1}}{{\rm 1}+e^{-\tilde{d}\left(\min \right)} } , 
\end{equation} 
 where $\tilde{d}\left(\min \right)$ is evaluated as 
\begin{equation} \label{eq35} 
\tilde{d}\left(\min \right)=\mathop{\min }\limits_{\forall k} {\left(d_{\nu ,k} \left(f\right)-\left(d_{k,\nu } \left(g\right)+d_{k,\nu } \left(c\right)\right)\right)}. 
\end{equation} 
 Using \eqref{eq35}, apply Sub-procedure 3. To evaluate the ${\rm {\mathcal{D}}}_{P} \left(\nu ,{\rm {\mathcal{A}}}\right)$ relations between $\nu $ and the elements of ${\rm {\mathcal{A}}}$ at ${\rm {\mathcal{D}}}_{P} \left(\xi ,\nu \right)=\xi \angle \nu $, apply Sub-procedure 4.

Step 4. Apply Steps 2-3, until $i<N_{it} $, where $i$ is the actual iteration, $N_{it} $ is the total number of iterations.  
\end{algo}

The related steps are detailed in Sub-procedures 1-4. In Step 3 of Sub-procedure 1, the best ${{\mathcal{A}}_{s}}$ solutions refer to those solutions from $\mathcal{A}$ that have the largest values of the crowding distance \cite{ref27}. Particularly, in this step, the solutions are also sorted and compared by a crowded comparison operator to find the best solution. 

\setcounter{algocf}{0}
\begin{subproc}
  \DontPrintSemicolon
\caption{\textit{}} 
\textbf{Step 1}. Set $\xi =\nu $, and add $\nu $ to ${\rm {\mathcal{A}}}$.

\textbf{Step 2}. If $\left|{\rm {\mathcal{A}}}\right|>A_{s} $, where $\left|{\rm {\mathcal{A}}}\right|$ is the number of elements in ${\rm {\mathcal{A}}}$, $A_{s} $ is the maximal archive size, then assign $\Delta _{cr} \left({\rm {\mathcal{A}}}\right)$ to ${\rm {\mathcal{A}}}$, where $\Delta _{cr} \left(\cdot \right)$ is the crowding distance.

\textbf{Step 3}. Select the best $A_{s} $ elements.  
\end{subproc}

\begin{subproc}
  \DontPrintSemicolon
\caption{\textit{}} 
\textbf{Step 1}. Set $\xi =\nu $, and add $\nu $ to ${\rm {\mathcal{A}}}$.

\textbf{Step 2}. Remove all the $k$ dominated points from ${\rm {\mathcal{A}}}$.  
\end{subproc}

\begin{subproc}
  \DontPrintSemicolon
\caption{\textit{}} 
\textbf{Step 1}. Set $\xi =k_{\tilde{d}\left(\min \right)} $, where $k_{\tilde{d}\left(\min \right)} $ is a point of ${\rm {\mathcal{A}}}$ that corresponds to $\tilde{d}\left(\min \right)$ (see \eqref{eq35}) with probability ${\Pr }\left(\left. \xi =\nu \right|\xi \angle \nu ,\nu \angle \left({\rm {\mathcal{A}}}\right)_{k} \right)$ (see \eqref{eq34}).

\textbf{Step 2}. Otherwise set $\xi =\nu $.  
\end{subproc}

\begin{subproc}
  \DontPrintSemicolon
\caption{\textit{}} 
\textbf{Step 1}. If ${\rm {\mathcal{D}}}_{P} \left(\nu ,{\rm {\mathcal{A}}}\right)={\rm {\mathcal{A}}}\neg \angle \nu $, i.e., $\nu $ is non-dominating with respect to ${\rm {\mathcal{A}}}$, then set $\xi =\nu $, and add $\nu $ to ${\rm {\mathcal{A}}}$. If $\left|{\rm {\mathcal{A}}}\right|>A_{s} $, then assign $\Delta _{cr} \left({\rm {\mathcal{A}}}\right)$ to ${\rm {\mathcal{A}}}$, and select the best $A_{s} $ elements.

\textbf{Step 2}. If ${\rm {\mathcal{D}}}_{P} \left(\nu ,\left({\rm {\mathcal{A}}}\right)_{k} \right)=\left({\rm {\mathcal{A}}}\right)_{k} \angle \nu $, i.e., $\nu $ dominates $k$ points in ${\rm {\mathcal{A}}}$, then set $\xi =\nu $, and add $\nu $ to ${\rm {\mathcal{A}}}$. Remove the $k$ points from ${\rm {\mathcal{A}}}$.  
\end{subproc}

\end{proof} 

\subsubsection{Computational Complexity of QTAM}
Following the complexity analysis of \cite{ref24, ref25, ref26, ref27}, the computational complexity of QTAM is evaluated as 
\begin{equation} \label{eq36} 
{\rm \mathcal{O}}\left(N_{d} N_{it} \left|{\rm {\mathcal{P}}}\right|\left(N_{obj} +\log _{2} \left(\left|{\rm {\mathcal{P}}}\right|\right)\right)\right), 
\end{equation} 
 where $N_{d} $ is the number of dominance measures, $N_{it} $ is the number of total iterations, $\left|{\rm {\mathcal{P}}}\right|$ is the population size, while $N_{obj} $ is the number of objectives.

\section{Wiring Optimization and Objective Function Maximization}
\label{sec4}
\subsection{Multilayer Quantum Circuit Grid}
 An $i$-th quantum gate of $QG$ is denoted by $g_{i} $, a $k$-th port of the quantum gate $g_{i} $ is referred to as $g_{i,k} $. Due to the hardware restrictions of gate-model quantum computer implementations \cite{ref16, ref17, ref18, ref19}, the quantum gates are applied in several rounds. Thus, a multilayer, $k$-dimensional (for simplicity we assume $k=2$), $n$-sized finite square-lattice grid $G_{QG}^{k,r} $ can be constructed for $QG$, where $r$ is the number of layers, $l_{z} $, $z=1,\ldots ,r$ . A quantum gate $g_{i} $ in the $z$-th layer $l_{z} $ is referred to as $g_{i}^{l_{z} } $, while a $k$-th port of $g_{i}^{l_{z} } $ is referred to as $g_{i,k}^{l_{z} } $.

\subsection{Method}

\begin{theorem}
There exists a method for the parallel optimization of quantum wiring in physical-layout of the quantum circuit and for the maximization of an objective function $C_{\alpha } \left(z\right)$.
\end{theorem}
\begin{proof}
The aim of this procedure (Method 1) is to provide a simultaneous physical-layer optimization and Hamiltonian minimization via the minimization of the wiring lengths in the multilayer structure of $QG$ and the maximization of the objective function (see also \sref{A1}). Formally, the aim of Method 1 is the $F_{{\rm 2}} \wedge F_{{\rm 3}} $ simultaneous realization of the objective functions $F_{{\rm 2}} $ and $F_{{\rm 3}} $.

Using the $G_{QG}^{k,r} $ multilayer grid of the $QG$ quantum circuit determined via $F_{{\rm 1}} $ and $F_{{\rm 2}} $, the aim of $F_{{\rm 3}} $ maximization of the objective function $C\left(z\right)$, where $z=z_{{\rm 1}} \ldots z_{n} $ in an $n$-length input string, where each $z_{i} $ is associated to an edge of $G_{QG}^{k,r} $ connecting two quantum ports. The objective function $C\left(z\right)$ associated to an arbitrary computational problem is defined as 
\begin{equation} \label{eq38} 
C\left(z\right)=\sum _{\left\langle i,j\right\rangle \in G_{QG}^{k,r} } C_{\left\langle i,j\right\rangle } \left(z\right), 
\end{equation} 
 where $C_{\left\langle i,j\right\rangle } $ is the objective function for an edge of $G_{QG}^{k,r} $ that connects quantum ports $i$ and $j$.

The $C^{{\rm *}} \left(z\right)$ maximization of objective function \eqref{eq38} yields a system state $\Psi $ for the quantum computer \cite{ref16, ref17, ref18, ref19} as 
\begin{equation} \label{eq39} 
\Psi =\left\langle \left. \gamma ,\mu ,C^{{\rm *}} \left(z\right)\right|\right. C^{{\rm *}} \left(z\right)\left| \gamma ,\mu ,C^{{\rm *}} \left(z\right)\right\rangle  , 
\end{equation} 
 where 
\begin{equation} \label{eq40} 
{\left| \gamma ,\mu ,C^{*} \left(z\right) \right\rangle} =U\left(B,\mu \right)U\left(C^{*} \left(z\right),\gamma \right){\left| s \right\rangle} , 
\end{equation} 
while 
\begin{equation} \label{eq43} 
U\left(C^{{\rm *}} \left(z\right),\gamma \right)\left| z\right\rangle  =e^{-i\gamma C^{{\rm *}} \left(z\right)} \left| z\right\rangle  , 
\end{equation} 
 where $\gamma $ is a single parameter \cite{ref16, ref17, ref18, ref19}.

The objective function \eqref{eq38} without loss of generality can be rewritten as 
\begin{equation} \label{eq44} 
C\left(z\right)=\sum _{\alpha } C_{\alpha } \left(z\right), 
\end{equation} 
 where $C_{\alpha } $ each act on a subset of bits, such that $C_{\alpha } \in \left\{{\rm 0,1}\right\}$. Therefore, there exists a selection of parameters of $\vec{\Phi }$ in \eqref{eq9} such that \eqref{eq44} picks up a maximized value $C^{{\rm *}} \left(z\right)$, which yields system state $\Upsilon $ as 
\begin{equation} \label{eq45} 
\Upsilon =\langle \vec{\Phi }|C^{{\rm *}} (z)|\vec{\Phi }\rangle . 
\end{equation} 
 Therefore, the resulting Hamiltonian $H$ associated to the system state \eqref{eq45} is minimized via $F_{{\rm 2}} $ (see \eqref{eq57}) as 
\begin{equation} \label{eq46} 
E_{L} (\vec{\Phi })=\min {\langle \vec{\Phi }|H|\vec{\Phi }\rangle }, 
\end{equation} 
 since the physical-layer optimization minimizes the $\ell _{ij} $ physical distance between the quantum ports, therefore the energy $E_{L} (\vec{\Phi })$ of the Hamiltonian associated to $\vec{\Phi }$ is reduced to a minima. 

The steps of the method $F_{{\rm 2}} \wedge F_{{\rm 3}} $ are given in Method 1. The method minimizes the number of quantum wires in the physical-layout of $QG$, and also achieves the desired system state $\Psi $ of \eqref{eq39}. 

\setcounter{algocf}{0}
\begin{proced}
  \DontPrintSemicolon
\caption{\textit{Quantum Wiring Optimization and Objective Function Maximization}}
\textbf{Step 1}. Construct the $G_{QG}^{k,r} $ multilayer grid of the $QG$ quantum circuit, with $r$ layers $l_{{\rm 1}} ,\ldots ,l_{r} $. Determine the 
\[C\left(z\right)=\sum _{\left\langle i,j\right\rangle \in G_{QG}^{k,r} } C_{\left\langle i,j\right\rangle } \left(z\right)\] 
objective function, where each $C_{\left\langle i,j\right\rangle } $ refers to the objective function for an edge in $G_{QG}^{k,r} $ connecting quantum ports $i$ and $j$, defined as 
\[C_{\left\langle i,j\right\rangle } \left(z\right)=\frac{{\rm 1}}{{\rm 2}} \left({\rm 1}-z_{i} z_{j} \right),\] 
where $z_{i} =\pm 1$.

\textbf{Step 2}. Find the optimal assignment of separation point $\Delta $ in $G_{QG}^{k,r} =\left(V,E,f\right)$ at a physical-layer blockage $\beta $ via a minimum-cost tree in $G_{QG}^{k,r} $ containing at least one port from each quantum gate $g_{i} $, $i=1,\ldots ,\left|V\right|$. For all pairs of quantum gates $g_{i} $, $g_{j} $, minimize the $f_{p,c} $ path cost (${\rm L1}$ distance) between a source quantum gate $g_{i} $ and destination quantum gate $g_{j} $ and then maximize the overlapped ${\rm L1}$ distance between $g_{i} $ and $\Delta $.

\textbf{Step 3}. For the $s$ found assignments of $\Delta $ in Step 2, evaluate the objective functions $C_{\alpha _{i} } $, $k=1,\ldots ,s$, where $C_{\alpha _{0} } $ is the initial value. Let the two paths ${\rm {\mathcal{P}}}_{{\rm 1}} $ and ${\rm {\mathcal{P}}}_{{\rm 2}} $ between quantum ports $g_{i{\rm ,1}} $, $g_{j{\rm ,1}} $, $g_{j{\rm ,2}} $ be given as ${\rm {\mathcal{P}}}_{{\rm 1}} :g_{i{\rm ,1}} \to \Delta \to g_{j{\rm ,1}} $, and ${\rm {\mathcal{P}}}_{{\rm 2}} :g_{i{\rm ,1}} \to \Delta \to g_{j{\rm ,2}} $. Evaluate objective functions $C_{\left\langle g_{i{\rm ,1}} ,\Delta \right\rangle } \left(z\right)$, $C_{\left\langle \Delta ,g_{j{\rm ,1}} \right\rangle } \left(z\right)$ and $C_{\left\langle \Delta ,g_{j{\rm ,2}} \right\rangle } \left(z\right)$.

\textbf{Step 4}. Select that $k$-th solution, for which 
\[{{C}_{{{\alpha }_{k}}}}\left( z \right)=C_{\left\langle {{g}_{i,1}},\Delta  \right\rangle }^{\left( k \right)}\left( z \right)+C_{\left\langle \Delta ,{{g}_{j,1}} \right\rangle }^{\left( k \right)}\left( z \right)+C_{\left\langle \Delta ,{{g}_{j,2}} \right\rangle }^{\left( k \right)}\left( z \right)\]
is maximal, where $C_{\left\langle i,j\right\rangle }^{\left(k\right)} $ is the objective function associated to a $k$-th solution between quantum ports $g_{i{\rm ,1}} $, $g_{j{\rm ,1}} $, and $g_{i{\rm ,1}} $, $g_{j{\rm ,2}} $ in $G_{QG}^{k,r} $. The resulting $C_{\alpha }^{{\rm *}} \left(z\right)$ for ${\rm {\mathcal{P}}}_{{\rm 1}} $ and ${\rm {\mathcal{P}}}_{{\rm 2}} $ is as 
\[C_{\alpha }^{*} \left(z\right)=\mathop{\mathop{\max }}\limits_{k} {\kern 1pt} \left(C_{\alpha _{k} } \left(z\right)\right).\] 

\textbf{Step 5}. Repeat steps 2-4 for all paths of $G_{QG}^{k,r} $.
\end{proced}

\end{proof}

The steps of Method 1 are illustrated in \fref{fig2}, using the $G_{QG}^{k,r} $ multilayer topology of the $QG$ quantum gate structure, $l_{i} $ refers to the $i$-th layer of $G_{QG}^{k,r} $.

 \begin{center}
\begin{figure}[h!]
%\vspace{-0.4cm}
\begin{center}
\includegraphics[angle = 0,width=0.8\linewidth]{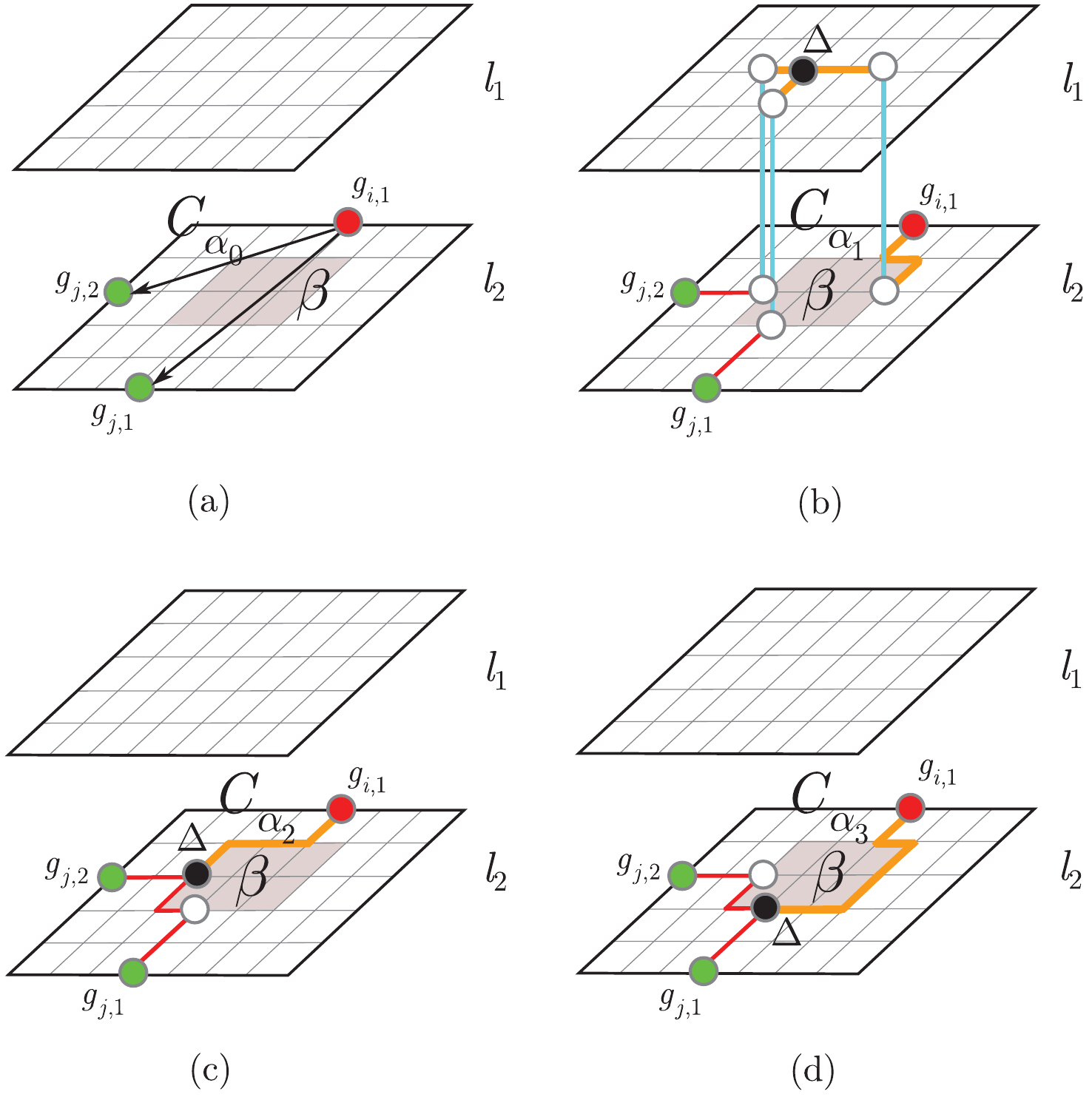}
\caption{The aim is to find the optimal wiring in $G_{QG}^{k,r} $ for the $QG$ quantum circuit (minimal path length with maximal overlapped path between $g_{i{\rm ,1}} $ and $g_{j{\rm ,1}} $,$g_{j{\rm ,2}} $) such that the $C_{\alpha } $ objective function associated to the paths ${\rm {\mathcal{P}}}_{{\rm 1}} :g_{i{\rm ,1}} \to g_{j{\rm ,1}} $, and ${\rm {\mathcal{P}}}_{{\rm 2}} :g_{i{\rm ,1}} \to g_{j{\rm ,2}} $ is maximal. (a): The initial objective function value is $C_{\alpha _{0} } $. A physical-layer blockage $\beta $ in the quantum circuit allows no to use paths ${\rm {\mathcal{P}}}_{{\rm 1}} $ and ${\rm {\mathcal{P}}}_{{\rm 2}} $. (b): The wire length is optimized via the selection point $\Delta $. The path cost is $f_{p,c} =11+3f_{l} $, where $f_{l} $ is the cost function of the path between the layers $l_{{\rm 1}} $ and $l_{{\rm 2}} $ (depicted by the blue vertical line), the path overlap from $g_{i{\rm ,1}} $ to $\Delta $ is $\tau _{o} =5+f_{l} $. The objective function value is $C_{\alpha _{{\rm 1}} } $. (c): The path cost is $f_{p,c} =10$, the path overlap from $g_{i{\rm ,1}} $ to $\Delta $ is $\tau _{o} =4$. The objective function value is $C_{\alpha _{{\rm 2}} } $. (d): The path cost is $f_{p,c} =12$, the path overlap from $g_{i{\rm ,1}} $ to $\Delta $ is $\tau _{o} =6$. The objective function value is $C_{\alpha _{{\rm 3}} } $. The selected connection topology from (b), (c), and (d) is that which yields the maximized objective function $C_{\alpha }^{{\rm *}} $.} 
 \label{fig2}
 \end{center}
\end{figure}
\end{center}

\subsection{Quantum Circuit Minimization}
 For objective function $F_{{\rm 1}} $, the area minimization of the $QG$ quantum circuit requires the following constraints. Let $S_{v} \left(P_{i} \right)$ be the vertical symmetry axis of a proximity group $P_{i} $ \cite{ref24, ref25, ref26} on $QG$, and let $x_{S_{v} \left(P_{i} \right)} $ refer to the $x$-coordinate of $S_{v} \left(P_{i} \right)$. Then, by some symmetry considerations for $x_{S_{v} \left(P_{i} \right)} $, 
\begin{equation} \label{eq47} 
x_{S_{v} \left(P_{i} \right)} =\frac{{\rm 1}}{{\rm 2}} \left(x_{i}^{{\rm 1}} +x_{i}^{{\rm 2}} +\kappa _{i} \right), 
\end{equation} 
 where $x_{i} $ is the bottom-left $x$ coordinate of a cell $\sigma _{i} $, $\kappa _{i} $ is the width of $\sigma _{i} $, and 
\begin{equation} \label{eq48} 
y_{i}^{{\rm 1}} +\frac{h_{i} }{{\rm 2}} =y_{i}^{{\rm 2}} +\frac{h_{i} }{{\rm 2}} , 
\end{equation} 
 where $y_{i} $ is the bottom-left $y$ coordinate of a cell $\sigma _{i} $, $h_{i} $ is the height of $\sigma _{i} $.

Let $\left(\sigma ^{{\rm 1}} ,\sigma ^{{\rm 2}} \right)$ be a symmetry pair \cite{ref24, ref25, ref26} that refers to two matched cells placed symmetrically in relation to $S_{v} \left(P_{i} \right)$, with bottom-left coordinates $\left(\sigma ^{{\rm 1}} ,\sigma ^{{\rm 2}} \right)=\left(\left(x_{i}^{{\rm 1}} ,y_{i}^{{\rm 1}} \right),\left(x_{i}^{{\rm 2}} ,y_{i}^{{\rm 2}} \right)\right)$. Then, $x_{S_{v} \left(P_{i} \right)} $ can be rewritten as 
\begin{equation} \label{eq49} 
x_{S_{v} \left(P_{i} \right)} =x_{i}^{{\rm 1}} -x_{i} =x_{i}^{{\rm 2}} +x_{i} +\kappa _{i} , 
\end{equation} 
 with the relation $y_{i}^{{\rm 1}} =y_{i}^{{\rm 2}} =y_{i} $.

Let $\sigma ^{S} =\left(x_{i}^{S} ,y_{i}^{S} \right)$ be a cell which is placed centered \cite{ref24, ref25, ref26} with respect to $S_{v} \left(P_{i} \right)$. Then, $x_{S_{v} \left(P_{i} \right)} $ can be evaluated as 
\begin{equation} \label{eq50} 
x_{S_{v} \left(P_{i} \right)} =x_{i}^{S} +\frac{\kappa _{i} }{{\rm 2}} , 
\end{equation} 
 along with $y_{i}^{S} =y_{i} $. Note that it is also possible that for some cells in $QG$ there is no symmetry requirements, these cells are denoted by $\sigma ^{0} $.

As can be concluded, using objective function $F_{{\rm 1}} $ for the physical-layer minimization of $QG$, a $d$-dimensional constraint vector ${\rm \mathbf{x}}_{F_{{\rm 1}} }^{d} $ can be formulated with the symmetry considerations as follows: 
\begin{equation} \label{eq51} 
{\rm \mathbf{x}}_{F_{{\rm 1}} }^{d} =\sum _{N_{\left(\sigma ^{{\rm 1}} ,\sigma ^{{\rm 2}} \right)} } \left(x_{i} ,y_{i} ,r_{i} \right)+\sum _{N_{\sigma ^{S} } } \left(y_{i} ,r_{i} \right)+\sum _{N_{\sigma ^{0} } } \left(x_{i} ,y_{i} ,r_{i} \right), 
\end{equation} 
 where $N_{\left(\sigma ^{{\rm 1}} ,\sigma ^{{\rm 2}} \right)} $ is the number of $\left(\sigma ^{{\rm 1}} ,\sigma ^{{\rm 2}} \right)$ symmetry pairs, $N_{\sigma ^{S} } $ is the number of $\sigma ^{S} $-type cells, while $N_{\sigma ^{0} } $ is the number of $\sigma ^{0} $-type cells, while $r_{i} $ is the rotation angle of an $i$-th cell $\sigma _{i} $, respectively.

\subsubsection{Quantum Wire Area Minimization}
Objective function $F_{{\rm 2}} $ provides a minimization of the total quantum wire length of the $QG$ circuit. To achieve it we define a procedure that yields the minimized total quantum wire area, $w_{QG} $, of $QG$ as given by \eqref{eq7}. Let $\delta _{ij} $ be the effective width of the quantum wire $ij$ in the $QG$ circuit, defined as 
\begin{equation} \label{eq52} 
\delta _{ij} =\frac{\psi _{ij} }{J_{\max } \left(T_{ref} \right)h_{nom} } , 
\end{equation} 
 where $\psi _{ij} $ is the (root mean square) condensate wave function amplitude, $J_{\max } \left(T_{ref} \right)$ is the maximum allowed current density at a given reference temperature $T_{ref} $, while $h_{nom} $ is the nominal layer height. Since drops in the condensate wave function phase $\varphi _{ij} $ are also could present in the $QG$ circuit environment, the $\delta '_{ij} $ effective width of the quantum wire $ij$ can be rewritten as 
\begin{equation} \label{eq53} 
\delta '_{ij} =\frac{\psi _{ij} \ell _{eff} r_{0} \left(T_{ref} \right)}{\chi _{\varphi _{ij} } } , 
\end{equation} 
 where $\chi _{\varphi _{ij} } $ is a maximally allowed value for the phase drops, $\ell _{eff} $ is the effective length of the quantum wire, $\ell _{eff} \le \left(\chi _{\varphi _{ij} } \delta _{ij} \right)/\psi _{ij} r_{0} \left(T_{ref} \right),$ while $r_{0} \left(T_{ref} \right)$ is a conductor sheet resistance \cite{ref1, ref2, ref3, ref4, ref5}.

In a $G_{QG}^{k,r} $ multilayer topological representation of $QG$, the $\ell _{ij} $ distance between the quantum ports is as 
\begin{equation} \label{eq54} 
\ell _{ij} =\left|x_{i} -x_{j} \right|+\left|y_{i} -y_{j} \right|+\left|z_{i} -z_{j} \right|f_{l} , 
\end{equation} 
 where $f_{l} $ is a cost function between the layers of the multilayer structure of $QG$.

During the evaluation, let $w_{QG} \left(k\right)$ be the total quantum wire area of a particular net $k$ of the $QG$ circuit, 
\begin{equation} \label{eq55} 
w_{QG} \left(k\right)=\sum _{i=1}^{p} \sum _{j=1}^{q} \ell _{ij} \cdot \delta _{ij} \left(\psi _{ij} \right), 
\end{equation} 
 where $q$ quantum ports are considered as sources of condensate wave function amplitudes, while $p$ of $QG$ are sinks, thus \eqref{eq7} can be rewritten as 
\begin{equation} \label{eq56} 
F_{{\rm 2}} :w_{QG} =\min {\sum _{k=1}^{h} w_{QG} \left(k\right)}. 
\end{equation} 
 Since $\psi _{ij} $ is proportional to $\delta _{ij} \left(\psi _{ij} \right)$, \eqref{eq56} can be simplified as 
\begin{equation} \label{eq57} 
F_{{\rm 2}} :w'_{QG} =\min {\sum _{k=1}^{h} w'_{QG} \left(k\right)}, 
\end{equation} 
 where 
\begin{equation} \label{eq58} 
w'_{QG} \left(k\right)=\sum _{i=1}^{p} \sum _{j=1}^{q} \ell _{ij} \cdot \psi _{ij} , 
\end{equation} 
 where $\ell _{ij} $ is given in \eqref{eq54}.

In all quantum ports of a particular net $k$ of $QG$, the source quantum ports are denoted by positive sign \cite{ref24, ref25, ref26} in the condensate wave function amplitude, $\psi _{ij} $ assigned to quantum wire $ij$ between quantum ports $i$ and $j$ , while the sink ports are depicted by negative sign in the condensate wave function amplitude, $-\psi _{ij} $ with respect to a quantum wire $ij$ between quantum ports $i$ and $j$.

Thus the aim of $w_{QG} \left(k\right)$ in \eqref{eq55} is to determine a set of port-to-port connections in the $QG$ quantum circuit, such that the number of long connections is reduced in a particular net $k$ of $QG$ as much as possible. The result in \eqref{eq56} is therefore extends these requirements for all nets of $QG$.

\paragraph{Wave Function Amplitudes}
With respect to a particular quantum wire $ij$ between quantum ports $i$ and $j$ of $QG$, let $\psi _{i\to j} $ refer to the condensate wave function amplitude in direction $i\to j$, and let $\psi _{j\to i} $ refer to the condensate wave function amplitude in direction $j\to i$ in the quantum circuit. Then, the let be $\phi _{ij} $ defined for the condensate wave function amplitudes of quantum wire $ij$ as 
\begin{equation} \label{eq59} 
\phi _{ij} =\min {\left(\left|\psi _{i\to j} \right|,\left|\psi _{j\to i} \right|\right)}, 
\end{equation} 
 with a residual condensate wave function amplitude 
\begin{equation} \label{eq60} 
\xi _{i\to j} =\phi _{ij} -\psi _{i\to j} , 
\end{equation} 
 where $\psi _{i\to j} $ is an actual amplitude in the forward direction $i\to j$. Thus, the maximum amount of condensate wave function amplitude injectable to of quantum wire $ij$ in the forward direction $i\to j$ at the presence of $\psi _{i\to j} $ is $\xi _{i\to j} $ (see \eqref{eq60}). The following relations holds for a backward direction, $j\to i$, for the decrement of a current wave function amplitude $\psi _{i\to j} $ as 
\begin{equation} \label{eq61} 
\bar{\xi }_{j\to i} =-\psi _{i\to j} , 
\end{equation} 
 with residual quantum wire length 
\begin{equation} \label{eq62} 
\Gamma _{j\to i} =-\delta _{ij} , 
\end{equation} 
 where $\delta _{ij} $ is given in \eqref{eq52}.

By some fundamental assumptions, the ${\rm {\mathcal{N}}}_{R} $ residual network of $QG$ is therefore a network of the quantum circuit with forward edges for the increment of the wave function amplitude $\psi $, and backward edges for the decrement of $\psi $. To avoid the problem of negative wire lengths the Bellman-Ford algorithm \cite{ref24, ref25, ref26} can be utilized in an iterative manner in the residual directed graph of the $QG$ topology.

To find a path between all pairs of quantum gates in the directed graph of the $QG$ quantum circuit, the directed graph has to be strongly connected. The strong-connectivity of the $h$ nets with the parallel minimization of the connections of the $QG$ topology can be achieved by a minimum spanning tree method such as Kruskal's algorithm \cite{ref24, ref25, ref26}.
 
\begin{lemma}
The objective function $F_{{\rm 2}} $ is feasible in a multilayer $QG$ quantum circuit structure.
\end{lemma}
\begin{proof}
 The procedure defined for the realization of objective function $F_{{\rm 2}} $ on a $QG$ quantum circuit is summarized in Method 2. The proof assumes a superconducting architecture.
 
\setcounter{algocf}{1}
\begin{proced}
  \DontPrintSemicolon
\caption{\textit{Implementation of Objective Function $F_2$}}
\textbf{Step 1}. Assign the $\psi _{ij} $ condensate wave function amplitudes for all $ij$ quantum wires of $QG$ via Sub-method 2.1.

\textbf{Step 2}. Determine the residual network of $QG$ via Sub-method 2.2.

\textbf{Step 3}. Achieve the strong connectivity of $QG$ via Sub-method 2.3.

\textbf{Step 4}. Output the $QG$ quantum circuit topology such that $w_{QG} $ \eqref{eq7} is minimized.  
\end{proced}

The sub-procedures of Method 2 are detailed in Sub-methods 2.1, 2.2 and 2.3.  

\setcounter{algocf}{0}
\begin{subproc2}
  \DontPrintSemicolon
\caption{}
\textbf{Step 1}. Create a ${\rm M}_{QG} $ multilayer topological map of the network ${\rm {\mathcal{N}}}$ of $QG$ with the quantum gates and ports.

\textbf{Step 2}. From ${\rm M}_{QG} $ determine the $L_{c} $ connection list of ${\rm {\mathcal{N}}}$ in $QG$.

\textbf{Step 3}. Determine the $\delta _{ij} $ the effective width of the quantum wire $ij$ via \eqref{eq52}, for $\forall ij$ wires.

\textbf{Step 4}. Determine $\phi _{ij} $ via \eqref{eq59} for all quantum wires $ij$ of the $QG$ circuit.

\textbf{Step 5}. For a $k$-th net of $QG$, assign the wave function amplitude values $\psi _{ij} $ to $\forall ij$ quantum wires such that $w_{QG} \left(k\right)$ in \eqref{eq55} is minimized, with quantum wire length $\ell _{ij} $ \eqref{eq54}.  
\end{subproc2}

\begin{subproc2}
  \DontPrintSemicolon
\caption{}
\textbf{Step 1}. Create a $\bar{{\rm M}}_{QG} $ multilayer topological map of the ${\rm {\mathcal{N}}}_{R} $ residual network of $QG$.

\textbf{Step 2}. From $\bar{{\rm M}}_{QG} $ determine the $\bar{L}_{c} $ connection list of the ${\rm {\mathcal{N}}}_{R} $ residual network of $QG$.

\textbf{Step 3}. For $\forall i\to j$ forward edges of $\bar{{\rm M}}_{QG} $ of ${\rm {\mathcal{N}}}_{R} $, compute the $\xi _{i\to j} $ residual condensate wave function amplitude \eqref{eq60}, and for $\forall j\to i$ backward edges of $\bar{{\rm M}}_{QG} $, compute the quantity $\bar{\xi }_{j\to i} $ via \eqref{eq61}.

\textbf{Step 4}. Compute the residual negative quantum wire length $\Gamma _{j\to i} $ via \eqref{eq62}, using $\delta _{ij} $ from \eqref{eq52}.

\textbf{Step 5}. Determine the $\bar{C}$ negative cycles in the ${\bar{{\rm M}}_{QG}} $ of the ${\rm {\mathcal{N}}}_{R} $ residual network of $QG$ via the ${\rm {\mathcal{A}}}_{BF} $ Bellman-Ford algorithm \cite{ref24, ref25, ref26}.

\textbf{Step 6}. If $N_{\bar{C}} >0$, where $N_{\bar{C}} $ is the number of $\bar{C}$ negative cycles in $\bar{{\rm M}}_{QG} $, then update the $\psi _{ij} $ wave function amplitudes of the quantum wires $ij$ in the to cancel out the negative cycles.

\textbf{Step 7}. Re-calculate the values of \eqref{eq60}, \eqref{eq61} and \eqref{eq62} for the residual edges of ${\rm {\mathcal{N}}}_{R} $.

\textbf{Step 8}. Repeat steps 5-7, until $N_{\bar{C}} >0$.  
\end{subproc2}

\begin{subproc2}
  \DontPrintSemicolon
\caption{}
\textbf{Step 1}. For an $i$-th $sn_{k,i} $ subnet of a net $k$ of the $QG$ quantum circuit, set the quantum wire length to zero, $\delta _{ij} =0$ between quantum ports $i$ and $j$, for all $\forall i$.

\textbf{Step 2}. Determine the $L{\rm 2}$ (Euclidean) distance between the quantum ports of the subnets $sn_{k,i} $ (from each quantum port of a subnet to each other quantum port of all remaining subnets \cite{ref24}).

\textbf{Step 3}. Weight the $\delta _{ij} >0$ non-zero quantum wire lengths by the calculated $L{\rm 2}$ distance between the connections of the subnets of the $QG$ quantum circuit \cite{ref1, ref2, ref3, ref4, ref5}, \cite{ref24, ref25, ref26}.

\textbf{Step 4}. Determine the minimum spanning tree ${\rm {\mathcal{T}}}_{QG} $ via the ${\rm {\mathcal{A}}}_{K} $ Kruskal algorithm \cite{ref24}.

\textbf{Step 5}. Determine the set $S_{{\rm {\mathcal{T}}}_{QG} } $ of quantum wires with $\delta _{ij} >0$ from ${\rm {\mathcal{T}}}_{QG} $. Calculate $\delta _{S_{{\rm {\mathcal{T}}}_{QG} } } =\max \left(\delta _{ij} ,\delta '_{ij} ,\delta _{0} \right)$, where $\delta _{0} $ is the minimum width can be manufactured, while $\delta _{ij} $ and $\delta '_{ij} $ are given in \eqref{eq52} and \eqref{eq53}.

\textbf{Step 6}. Add the quantum wires of $S_{{\rm {\mathcal{T}}}_{QG} } $ to the ${\rm M}_{QG} $ multilayer topological map of the network ${\rm {\mathcal{N}}}$ of $QG$.

\textbf{Step 7}. Repeat steps 4-6 for $\forall k$ nets of the $QG$ quantum circuit, until ${\rm M}_{QG} $ is not strongly connected.  
\end{subproc2}

These conclude the proof.
\end{proof}

\subsubsection{Processing in the Multilayer Structure}

The $G_{QG}^{k,z} $ grid consists of all $g_{i} $ quantum gates of $QG$ in a multilayer structure, such that the $g_{i,k}^{l_{z} } $ appropriate ports of the quantum gates are associated via an directed graph ${\rm {\rm G}}=\left(V,E,f_{c} \right)$, where $V$ is the set of ports, $g_{i,k}^{l_{z} } \subseteq V$, $E$ is the set of edges, and $f_{c} $ is a cost function, to achieve the gate-to-gate connectivity.

As a hardware restriction we use a constraint on the quantum gate structure, it is assumed in the model that a given quantum system cannot participate in more than one quantum gate at a particular time.

The distance in the rectilinear grid $G_{QG}^{k,z} $ of $QG$ is measured by the $d_{{\rm L1}} \left(\cdot \right)$ ${\rm L1}$-distance function. Between two network ports $x,y\in V$, $x=\left(j,k\right)$, $y=\left(m,o\right)$, $d_{{\rm L1}} \left(\cdot \right)$ is as 
\begin{equation} \label{eq63} 
d_{{\rm L1}} \left(x,y\right)=d_{{\rm L1}} \left(\left(j,k\right),\left(m,o\right)\right)=\left|m-j\right|+\left|o-k\right|. 
\end{equation} 
 The quantum port selection in the $G_{QG}^{k,r} $ multilayer structure of $QG$, with $r$ layers $l_{z} $, $z=1,\ldots ,r$, and $k=2$ dimension in each layers is illustrated in \fref{figA1}.

\begin{center}
\begin{figure*}[htbp]
%\vspace{-0.6cm}
\begin{center}
\includegraphics[angle = 0,width=1\linewidth]{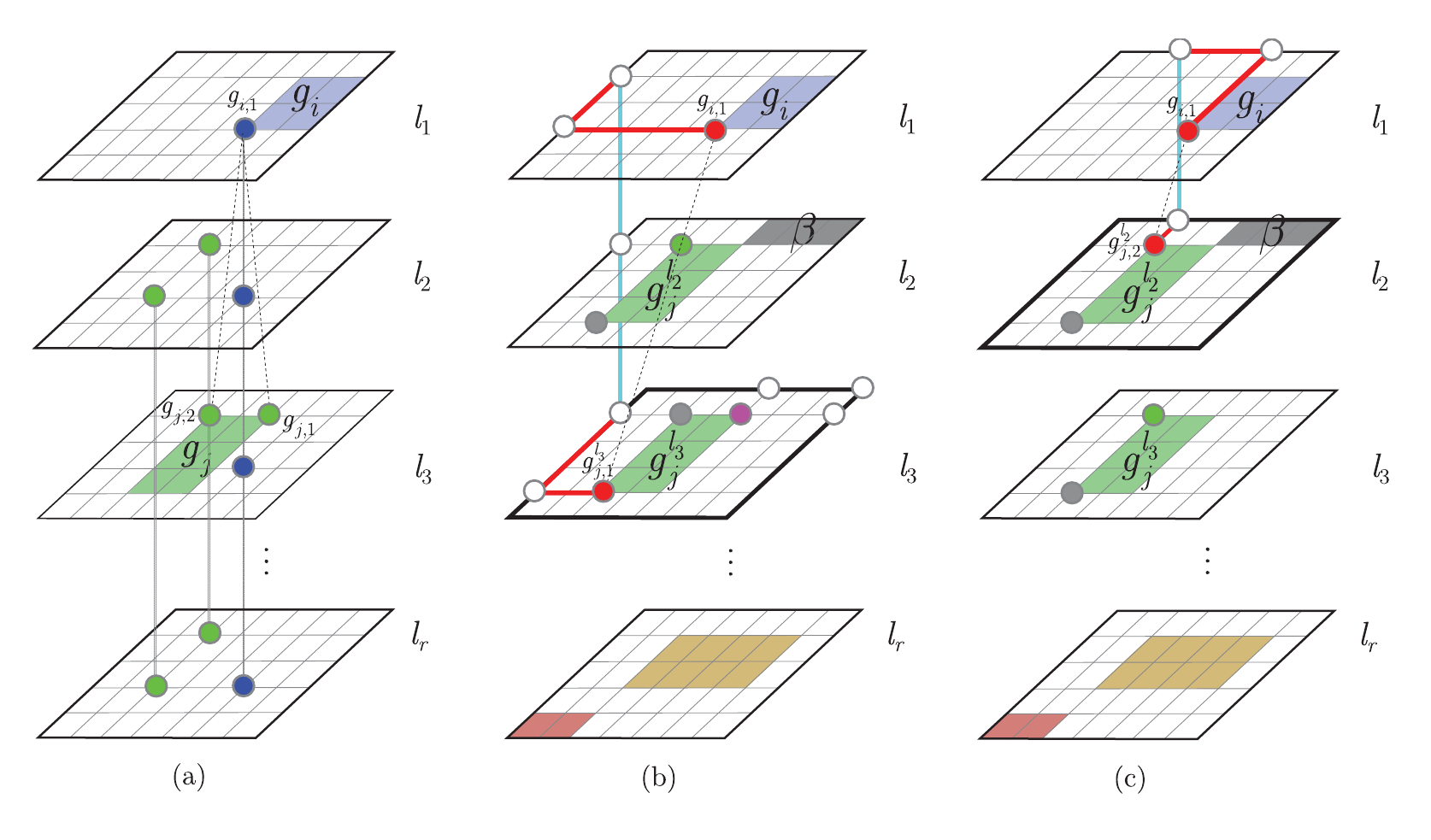}
\caption{The method of port allocation of the quantum gates in the $G_{QG}^{k,r} $ multilayer structure, with $r$ layers $l_{z} $, $z=1,\ldots ,r$, and $k=2$ dimension in each layers. The aim of the multiport selection is to find the shortest path between ports of quantum gates $g_{i} $ (blue rectangle) and $g_{j} $ (green rectangle) in the $G_{QG}^{{\rm 2,}r} $ multilayer structure. (a): The quantum ports needed to be connected in $QG$ are port $g_{i{\rm ,1}} $ in quantum gate $g_{i} $ in layer $l_{{\rm 1}} $, and ports $g_{j{\rm ,1}} $ and $g_{j{\rm ,2}} $ of quantum gate $g_{j} $ in layer $l_{{\rm 3}} $. (b): Due to a hardware restriction on quantum computers, the quantum gates are applied in several rounds in the different layers of the quantum circuit $QG$. Quantum gate $g_{j} $ is applied in two rounds in two different layers that is depicted $g_{j}^{l_{{\rm 3}} } $ and $g_{j}^{l_{{\rm 2}} } $. For the layer-$l_{{\rm 3}} $ quantum gate $g_{j}^{l_{{\rm 3}} } $, the active port is $g_{j{\rm ,1}}^{l_{{\rm 3}} } $ (red), while the other port is not accessible (gray) in $l_{{\rm 3}} $. The $g_{j{\rm ,1}}^{l_{{\rm 3}} } $ port, due to a physical-layer blockage $\beta $ in the quantum circuit of the above layer $l_{{\rm 2}} $ does not allow to minimize the path cost between ports $g_{i{\rm ,1}} $ and $g_{j{\rm ,1}}^{l_{{\rm 3}} } $. The target port $g_{j{\rm ,1}}^{l_{{\rm 3}} } $ is therefore referred to as a blocked port (depicted by pink), and a new port of is $g_{j}^{l_{{\rm 3}} } $ selected for $g_{j{\rm ,1}}^{l_{{\rm 3}} } $ (new port depicted by red). (c): For the layer-$l_{{\rm 2}} $ quantum gate $g_{j}^{l_{{\rm 2}} } $, the active port is $g_{j{\rm ,2}}^{l_{{\rm 2}} } $ (red), while the remaining port is not available (gray) in $l_{{\rm 2}} $. The white dots (vertices) represent auxiliary ports in the grid structure of the quantum circuit. In $G_{QG}^{{\rm 2,}r} $, each vertices could have a maximum of 8 neighbors, thus for a given port $g_{j,k} $ of a quantum gate $g_{j} $, ${\rm deg}\left(g_{j,k} \right)\le {\rm 8}$.} 
 \label{figA1}
 \end{center}
\end{figure*}
\end{center}

\paragraph{Algorithm}
\begin{theorem}
The Quantum Shortest Path Algorithm finds shortest paths in a multilayer $QG$ quantum circuit structure.
\end{theorem}
\begin{proof}
The steps of the shortest path determination between the ports of the quantum gates in a multilayer structure are included in Algorithm 2. 

\setcounter{algocf}{1}
\begin{algo}
  \DontPrintSemicolon
\caption{\textit{Quantum Shortest Path Algorithm (QSPA)}}
\textbf{Step 1}. Create the $G_{QG}^{k,r} $ multilayer structure of $QG$, with $r$ layers $l_{z} $, $z=1,\ldots ,r$, and $k$ dimension in each layers. From $G_{QG}^{k,r} $ generate a list $L_{{\rm {\mathcal{P}}}\in {\rm {\rm Q}{\rm G}}} $ of the paths between each start quantum gate port to each end quantum gate port in the $G_{QG}^{k,r} $ structure of $QG$ quantum circuit.

\textbf{Step 2}. Due to the hardware restrictions of quantum computers, add the decomposed quantum gate port information and its layer information to $L_{{\rm {\mathcal{P}}}\in {\rm {\rm Q}{\rm G}}} $. Add the $\beta $ physical-layer blockage information to $L_{{\rm {\mathcal{P}}}\in {\rm {\rm Q}{\rm G}}} $.

\textbf{Step 3}. For a quantum port pair $\left(x,y\right)\in G_{QG}^{k,r} $ define the $f_{c} \left(x,y\right)$ cost function, as 
\[f_{c} \left(x,y\right)=\gamma \left(x,y\right)+d_{{\rm L1}} \left(x,y\right),\] 
where $\gamma \left(x,y\right)$ is the real path size from $x$ to $y$ in the multilayer grid structure $G_{QG}^{k,r} $ of $QG$, while $d_{{\rm L1}} \left(x,y\right)$ is the ${\rm L1}$ distance in the grid structure as given by \eqref{eq63}.

\textbf{Step 4}. Using $L_{{\rm {\mathcal{P}}}\in {\rm {\rm Q}{\rm G}}} $ and cost function $f_{c} \left(x,y\right)$, apply the $A^{{\rm *}} $ parallel search \cite{ref24, ref25, ref26} to determine the lowest cost path ${\rm {\mathcal{P}}}^{{\rm *}} \left(x,y\right)$.  
\end{algo}

\end{proof}

\paragraph{Complexity Analysis}
The complexity analysis of Algorithm 2 is as follows. Since the QSPA algorithm (Algorithm 2) is based on the $A^{{\rm *}} $ search method \cite{ref24, ref25, ref26}, the complexity is trivially yielded by the complexity of the $A^{{\rm *}} $ search algorithm.

\section{Performance Evaluation}
\label{sec5}
In this section, we compare the performance of the proposed QTAM method with a multiobjective evolutionary algorithm called NSGA-II \cite{com1}. We selected this multiobjective evolutionary algorithm for the comparison, since the method can be adjusted for circuit designing. 

The computational complexity of NSGA-II is proven to be ${\rm {\mathcal O}}\left(N_{it} N_{obj} \left|{\rm {\mathcal P}}\right|^{2} \right)$ in general, while at an optimized nondominated procedure, the complexity can be reduced to ${\rm {\mathcal O}}\left(N_{it} N_{obj} \left|{\rm {\mathcal P}}\right|\log _{2} \left|{\rm {\mathcal P}}\right|\right)$. We take into consideration both situations for a comparison. The complexity of QTAM is given in \eqref{eq36}.

The complexity of the methods in terms of the number of iterations, $N_{O}$, is compared in \fref{figA2}. The performance of QTAM is depicted in \fref{figA2}(a), while \fref{figA2}(b) and \fref{figA2}(c) illustrate the performances of the NSGA-II and optimized NSGA-II, respectively.

For the comparison, the $N_{obj} $ parameter is set to $N_{obj} =5$, while for the QTAM method, $N_{d} $ is set to $N_{d} =3$. 

\begin{center}
\begin{figure*}[htbp]
%\vspace{-0.6cm}
\begin{center}
\includegraphics[angle = 0,width=1\linewidth]{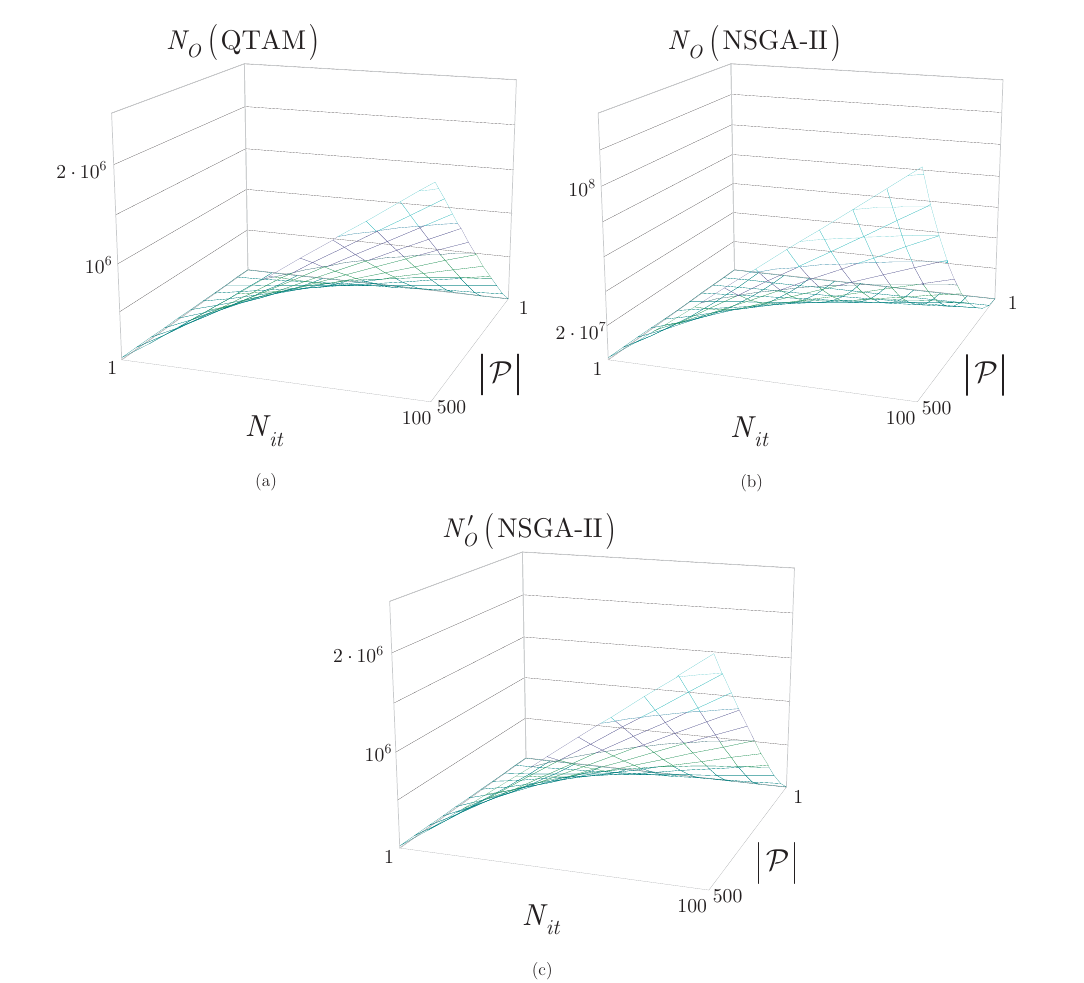}
\caption{(a): The computational complexity ($N_{O} $: number of operations) of QTAM in function of $N_{it} $ and $\left|{\rm {\mathcal P}}\right|$,  $N_{it} \in \left[1,100\right]$, $\left|{\rm {\mathcal P}}\right|\in \left[1,500\right]$. (b): The computational complexity of the NSGA-II method in function of $N_{it} $ and $\left|{\rm {\mathcal P}}\right|$,  $N_{it} \in \left[1,100\right]$, $\left|{\rm {\mathcal P}}\right|\in \left[1,500\right]$. (c): The computational complexity of the optimized NSGA-II in function of $N_{it} $ and $\left|{\rm {\mathcal P}}\right|$,  $N_{it} \in \left[1,100\right]$, $\left|{\rm {\mathcal P}}\right|\in \left[1,500\right]$.} 
 \label{figA2}
 \end{center}
\end{figure*}
\end{center}

In the analyzed range, the maximized values of $N_{O} $ are $N_{O} \left({\rm QTAM}\right)\approx 2\cdot 10^{6} $, $N_{O} (\text{NSGA-II})\approx 1.25\cdot 10^{8} $, and for the optimized NSGA-II scenario, $N'_{O} \left({\text{NSGA-II}}\right)\approx 2.25\cdot 10^{6} $, respectively. In comparison to NSGA-II, the complexity of QTAM is significantly lower. Note, while the performance of QTAM and the optimized NSGA-II is closer, QTAM requires no any optimization of the complexity of the nondominated procedure.

\section{Conclusions}
\label{sec6}
The algorithms and methods presented here provide a framework for quantum circuit designs for near term gate-model quantum computers. Since our aim was to define a scheme for present and future quantum computers, the developed algorithms and methods were tailored for arbitrary-dimensional quantum systems and arbitrary quantum hardware restrictions. We demonstrated the results through gate-model quantum computer architectures; however, due to the flexibility of the scheme, arbitrary implementations and input constraints can be integrated into the quantum circuit minimization. The objective function that is the subject of the maximization in the method can also be selected arbitrarily. This allows a flexible implementation to solve any computational problem for experimental quantum computers with arbitrary hardware restrictions and development constraints. 

%\section*{Statements}
%\subsection*{Ethics statement}
%This work did not involve any active collection of human data.
%\subsection*{Data accessibility statement}
%This work does not have any experimental data.
%\subsection*{Competing financial interests statement}
%We have no competing financial interests.
%\subsection*{Competing interests statement}
%We have no competing interests.
%\subsection*{Funding}
%No relevant funding. 
%\subsection*{Authors’ contributions}
%L.GY. designed the protocol and wrote the manuscript. L.GY. and S.I. analyzed the results. All authors reviewed the manuscript.

\section*{Acknowledgements}
This work was partially supported by the European Research Council through the Advanced Fellow Grant, in part by the Royal Society’s Wolfson Research Merit Award, in part by the Engineering and Physical Sciences Research Council under Grant EP/L018659/1, by the Hungarian Scientific Research Fund - OTKA K-112125 and in part by the Engineering and Physical Sciences Research Council under Grant EP/L018659/1.

\newpage
\appendix
\setcounter{table}{0}
\setcounter{figure}{0}
\setcounter{equation}{0}
\setcounter{algocf}{0}
\renewcommand{\thetable}{\Alph{section}.\arabic{table}}
\renewcommand{\thefigure}{\Alph{section}.\arabic{figure}}
\renewcommand{\theequation}{\Alph{section}.\arabic{equation}}
\renewcommand{\thealgocf}{\Alph{section}.\arabic{algocf}}

\setlength{\arrayrulewidth}{0.1mm}
\setlength{\tabcolsep}{5pt}
\renewcommand{\arraystretch}{1.5}
\section{Appendix}

\subsection{Abbreviations}
\begin{description}
\item[EDA] Electronic Design Automation
\item[IC] Integrated Circuit
\item[QG] Quantum Gate
\item[QSPA] Quantum Shortest Path Algorithm
\item[QTAM] Quantum Triple Annealing Minimization
\item[SA] Stimulated Annealing
\item[VLSI] Very-Large-Scale Integration
\end{description}

\subsection{Notations}
The notations of the manuscript are summarized in \tref{tab1}.
\begin{center}
\begin{longtable}{||l|p{4.5in}||}
\caption{Summary of notations.}
\label{tab1}
\endfirsthead
\endhead
\hline
\textit{Notation} & \textit{Description} \\ \hline
$d$  &  Dimension of the quantum system, $d=2$ for a qubit system.  \\ \hline 
$H$  &  Hamiltonian operator.  \\ \hline 
$QG$  &  Reduced quantum circuit.  \\ \hline 
$QG_{R} $  &  Reference (non-reduced) quantum circuit.  \\ \hline 
$P_{QG} $  &  Output distribution of $QG$.  \\ \hline 
$P_{QG_{R} } $  &  Output distribution of $QG_{R} $.  \\ \hline 
$n$  &  Number of input quantum systems  \\ \hline 
$F_{i} $  &  An $i$-th objective function.  \\ \hline 
$s$  &  A solution.  \\ \hline 
$f\left(s\right)$  &  The relative performance of a solution $s$.  \\ \hline 
$T$  &  Control parameter in the SA method.  \\ \hline 
$R$  &  Temperature decreasing rate in SA.  \\ \hline 
$t$  &  Iteration counter for the SA method.  \\ \hline 
$k$  &  Scaling factor.  \\ \hline 
$T_{\max } $  &  Initial temperature in the SA framework.  \\ \hline 
$d\left(f\right)$, $d\left(g\right)$, $d\left(c\right)$  &  Objective, constraint and distribution closeness domination functions.  \\ \hline 
$\tilde{d}\left(f\right)$, $\tilde{d}\left(g\right)$, $\tilde{d}\left(c\right)$  &  Average values of objective, constraint and distribution closeness domination functions.  \\ \hline 
${\rm \mathbf{x}}$  &  A vector of design variables.  \\ \hline 
$\alpha $  &  A vector of weights.  \\ \hline 
$N_{obj} $  &  Number of objectives in the optimization procedure.  \\ \hline 
$F_{s} $  &  A single-objective function.  \\ \hline 
$A_{QG} $  &  Quantum circuit area of the $QG$ quantum gate structure.  \\ \hline 
$H'_{QG} $  &  Optimal circuit height of $QG$.  \\ \hline 
$D'_{QG} $  &  Optimal depth of $QG$.  \\ \hline 
$w_{QG} $  &  Total quantum wire area of $QG$.  \\ \hline 
$h$  &  Number of nets of the $QG$ circuit.  \\ \hline 
$p$  &  Number of quantum ports of the $QG$ quantum circuit considered as sources of a condensate wave function amplitude.  \\ \hline 
$q$  &  Number of quantum portsf considered as sinks of a condensate wave function amplitude.  \\ \hline 
$\ell _{ij} $  &  Length of the quantum wire $ij$.  \\ \hline 
$\delta _{ij} $  &  Effective width of the quantum wire $ij$.  \\ \hline 
$\psi _{ij} $  &  Condensate wave function amplitude associated to the quantum wire $ij$.  \\ \hline 
$\vec{\Phi }$  &  A collection of $L$ parameters $\vec{\Phi }=\Phi _{{\rm 1}} ,\ldots ,\Phi _{L} $.  \\ \hline 
$| \vec{\Phi }\rangle  $  &  A system state of the quantum computer, $| \vec{\Phi }\rangle  =U_{L} \left(\Phi _{L} \right),\ldots ,U_{{\rm 1}} \left(\Phi _{{\rm 1}} \right)\left| \varphi \right\rangle  $, where $U_{i} $ is an $i$-th unitary that depends on a set of parameters $\Phi _{i} $, while $\left| \varphi \right\rangle  $ is an initial state.  \\ \hline 
$m$  &  Total number of measurements in the $M$ measurement block, $m=N_{M} \left|M\right|$, where $N_{M} $ is the number of measurement rounds, $\left|M\right|$ is the number of measurement gates in the $M$ measurement block.  \\ \hline 
$c_{s} \left(\cdot \right)$  &  Sum of distribution closeness violation values.  \\ \hline 
$D\left(\left. \cdot \right\| \cdot \right)$  &  Relative entropy function.  \\ \hline 
$d_{x,y} \left(c\right)$  &  Distribution closeness dominance function for solutions $x$ and $y$.  \\ \hline 
$v_{i}^{c} $  &  An $i$-th distribution closeness violation value.  \\ \hline 
$N_{v} $  &  The number of distribution closeness violation values for a solution $z$.  \\ \hline 
$g_{s} \left(\cdot \right)$  &  Sum of all constraint violation values.  \\ \hline 
$v_{i}^{g} $  &  An $i$-th constraint violation value.  \\ \hline 
$N_{g} $  &  Number of constraint violation values for a solution $z$.  \\ \hline 
$d_{x,y} \left(f\right)$  &  Objective dominance function.  \\ \hline 
$R_{i} $  &  Range of an objective $i$.  \\ \hline 
${\rm {\mathcal{A}}}$  &  Archive.  \\ \hline 
$\xi $  &  A random solution from ${\rm {\mathcal{A}}}$.  \\ \hline 
$\nu $  &  Parameter, $\nu =\Xi \left(\xi \right)$, where $\Xi \left(\cdot \right)$ is a moving operator.  \\ \hline 
${\rm {\mathcal{D}}}_{P} \left(\cdot \right)$  &  Constrained Pareto dominance checking function.  \\ \hline 
$\angle $  &  Pareto dominance operator; for $\nu \angle \xi $, $\xi $ dominates $\nu $.  \\ \hline 
$N_{it} $  &  Total number of iterations.  \\ \hline 
$N_{d} $  &  Number of dominance measures.  \\ \hline 
$\left|{\rm {\mathcal{P}}}\right|$  &  Population size.  \\ \hline 
$N_{obj} $  &  Number of objectives.  \\ \hline 
$g_{i} $  &  An $i$-th quantum gate of $QG$.  \\ \hline 
$g_{i,k} $  &  A $k$-th port of the quantum gate $g_{i} $ of the quantum circuit.  \\ \hline 
$G_{QG}^{k,r} $  &  A multilayer, $k$-dimensional $n$-sized finite square-lattice base-graph rectilinear grid, where $r$ is the number of layers, $l_{z} $, $z=1,\ldots ,r$.  \\ \hline 
$g_{i}^{l_{z} } $  &  A quantum gate $g_{i} $ in the $z$-th layer $l_{z} $ of $G_{QG}^{k,r} $.  \\ \hline 
$g_{i,k}^{l_{z} } $  &  A $k$-th port of $g_{i}^{l_{z} } $ in $G_{QG}^{k,r} $.  \\ \hline 
$C\left(z\right)$  &  Objective function of a computational problem, $z$ is a bitstring that encodes the state of the quantum circuit.  \\ \hline 
$C_{\left\langle i,j\right\rangle } $  &  Objective function for an edge of $G_{QG}^{k,r} $ that connects quantum ports $i$ and $j$.  \\ \hline 
$z_{i} $  &  Parameter, $z_{i} =\pm 1$.  \\ \hline 
$C^{{\rm *}} \left(z\right)$  &  Maximized objective function.  \\ \hline 
$U$  &  Unitary operation of the quantum computer.  \\ \hline 
$\sigma _{x} $  &  Pauli $X$-operator.  \\ \hline 
$\mu $  &  Control parameter.  \\ \hline 
$\gamma $  &  Single parameter.  \\ \hline 
$\ell _{ij} $  &  Distance between the quantum ports in $G_{QG}^{k,r} $.  \\ \hline 
$E_{L} \left(\vec{\Phi }\right)$  &  Energy $E_{L} \left(\vec{\Phi }\right)$ of the Hamiltonian at a system state $\vec{\Phi }$.  \\ \hline 
$\Delta $  &  Separation point in $G_{QG}^{k,r} $ of the quantum circuit.  \\ \hline 
$\beta $  &  A physical-layer blockage in the actual layer of the quantum circuit.  \\ \hline 
${\rm {\mathcal{P}}}$  &  A path between the quantum ports of the quantum circuit.  \\ \hline 
$S_{v} \left(P_{i} \right)$  &  A vertical symmetry axis of a proximity group $P_{i} $ on $QG$.  \\ \hline 
$x_{S_{v} \left(P_{i} \right)} $  &  The $x$-coordinate of $S_{v} \left(P_{i} \right)$.  \\ \hline 
$\sigma _{i} $  &  A cell in the grid of the quantum circuit.  \\ \hline 
$x_{i} $  &  A bottom-left $x$ coordinate of a cell $\sigma _{i} $ in the grid of the quantum circuit.  \\ \hline 
$\kappa _{i} $  &  Width of a cell $\sigma _{i} $ in the grid of the quantum circuit.  \\ \hline 
$\left(\sigma ^{{\rm 1}} ,\sigma ^{{\rm 2}} \right)$  &  Symmetry pair.  \\ \hline 
${\rm \mathbf{x}}_{F_{{\rm 1}} }^{d} $  &  A $d$-dimensional constraint vector with the symmetry considerations.  \\ \hline 
$N_{\left(\sigma ^{{\rm 1}} ,\sigma ^{{\rm 2}} \right)} $  &  Number of $\left(\sigma ^{{\rm 1}} ,\sigma ^{{\rm 2}} \right)$ symmetry pairs in $G_{QG}^{k,r} $.  \\ \hline 
$N_{\sigma ^{S} } $  &  Number of $\sigma ^{S} $-type cells in $G_{QG}^{k,r} $.  \\ \hline 
$N_{\sigma ^{0} } $  &  Number of $\sigma ^{0} $-type cells in $G_{QG}^{k,r} $.  \\ \hline 
$r_{i} $  &  Rotation angle of an $i$-th cell $\sigma _{i} $ in $G_{QG}^{k,r} $.  \\ \hline 
$\delta _{ij} $  &  Effective width of the quantum wire $ij$ in the $QG$ circuit.  \\ \hline 
$J_{\max } \left(T_{ref} \right)$  &  Maximum allowed current density at a given reference temperature $T_{ref} $.  \\ \hline 
$h_{nom} $  &  A nominal layer height.  \\ \hline 
$\delta '_{ij} $  &  Effective width of the quantum wire $ij$.  \\ \hline 
$\chi _{\varphi _{ij} } $  &  Maximally allowed value for the phase drops.  \\ \hline 
$\ell _{eff} $  &  Effective length of the quantum wire, $\ell _{eff} \le \left(\chi _{\varphi _{ij} } \delta _{ij} \right)/\psi _{ij} r_{0} \left(T_{ref} \right),$ where $r_{0} \left(T_{ref} \right)$ is a conductor sheet resistance.  \\ \hline 
$\ell _{ij} $  &  Distance between the quantum ports in $G_{QG}^{k,r} $, where $f_{l} $ is a cost function between the layers of the multilayer structure of $QG$.  \\ \hline 
$w_{QG} \left(k\right)$  &  Total quantum wire area of a particular net $k$ of the $QG$ circuit.  \\ \hline 
$\psi _{i\to j} $  &  Condensate wave function amplitude in direction $i\to j$ between the quantum ports.  \\ \hline 
$\psi _{j\to i} $  &  Condensate wave function amplitude in direction $j\to i$ between the quantum ports.  \\ \hline 
$\xi _{i\to j} $  &  Residual condensate wave function amplitude, $\xi _{i\to j} =\phi _{ij} -\psi _{i\to j} $, where $\phi _{ij} =\min \left(\left|\psi _{i\to j} \right|,\left|\psi _{j\to i} \right|\right)$.  \\ \hline 
$\bar{\xi }_{j\to i} $  &  Decrement of a current wave function amplitude $\psi _{i\to j} $ for a backward direction, $j\to i$, $\bar{\xi }_{j\to i} =-\psi _{i\to j} $.  \\ \hline 
$\Gamma _{j\to i} $  &  A residual quantum wire length for $\bar{\xi }_{j\to i} $.  \\ \hline 
${\rm {\mathcal{N}}}$  &  Network of $QG$ quantum circuit.  \\ \hline 
${\rm {\mathcal{N}}}_{R} $  &  Residual network of $QG$ quantum circuit.  \\ \hline 
${\rm M}_{QG} $  &  Topological map of the network ${\rm {\mathcal{N}}}$.  \\ \hline 
$L_{c} $  &  Connection list of ${\rm {\mathcal{N}}}$ in $QG$.  \\ \hline 
$\bar{{\rm M}}_{QG} $  &  Topological map of the ${\rm {\mathcal{N}}}_{R} $ residual network of $QG$.  \\ \hline 
$\bar{L}_{c} $  &  Connection list of the ${\rm {\mathcal{N}}}_{R} $ residual network of $QG$.  \\ \hline 
$N_{\bar{C}} $  &  Number of $\bar{C}$ negative cycles in $\bar{{\rm M}}_{QG} $.  \\ \hline 
$sn_{k,i} $  &  An $i$-th subnet of a net $k$ of the $QG$ quantum circuit.  \\ \hline 
${\rm {\mathcal{A}}}_{BF} $  &  Bellman-Ford algorithm.  \\ \hline 
${\rm {\mathcal{A}}}_{K} $  &  Kruskal algorithm.  \\ \hline 
${\rm {\mathcal{T}}}_{QG} $  &  Minimum spanning tree.  \\ \hline 
$S_{{\rm {\mathcal{T}}}_{QG} } $  &  Set of quantum wires with $\delta _{ij} >0$.  \\ \hline 
$\delta _{0} $  &  Minimum width can be manufactured physically.  \\ \hline 
$d_{{\rm L1}} \left(\cdot \right)$  &  ${\rm L1}$-distance function.  \\ \hline 
$f_{c} \left(x,y\right)$  &  A cost function, for a quantum port pair $\left(x,y\right)\in G_{QG}^{k,r} $, defined as $f_{c} \left(x,y\right)=\gamma \left(x,y\right)+d_{{\rm L1}} \left(x,y\right)$, where $\gamma \left(x,y\right)$ is the real path size from $x$ to $y$ in the multilayer grid structure $G_{QG}^{k,r} $ of $QG$, while $d_{{\rm L1}} \left(x,y\right)$ is the ${\rm L1}$ distance in the grid structure.  \\ \hline 
$A^{{\rm *}} $  &  $A^{{\rm *}} $ search algorithm.  \\ \hline 
${\rm {\mathcal{P}}}^{{\rm *}} \left(x,y\right)$  &  A lowest cost path between quantum ports $\left(x,y\right)\in G_{QG}^{k,r} $.  \\ \hline
\end{longtable}
\end{center}

\begin{thebibliography}{10}
\bibitem{refpr} Preskill, J. Quantum Computing in the NISQ era and beyond, \textit{Quantum} 2, 79 (2018).

\bibitem{refha} Harrow, A. W. and Montanaro, A. Quantum Computational Supremacy, \textit{Nature}, vol 549, pages 203-209 (2017).

\bibitem{aar} Aaronson, S. and Chen, L. Complexity-theoretic foundations of quantum supremacy experiments. \textit{Proceedings of the 32nd Computational Complexity Conference}, CCC '17, pages 22:1-22:67, (2017).

\bibitem{ref1} Ofek, N. et al. Extending the lifetime of a quantum bit with error correction in superconducting circuits. \textit{Nature} 536, 441-445 (2016).

\bibitem{ref2} Debnath, S. et al. Demonstration of a small programmable quantum computer with atomic qubits. \textit{Nature} 536, 63-66 (2016).

\bibitem{ref3} Barends, R. et al. Superconducting quantum circuits at the surface code threshold for fault tolerance. \textit{Nature} 508, 500-503 (2014).

\bibitem{ref4} Monz, T. et al. Realization of a scalable Shor algorithm. \textit{Science} 351, 1068-1070 (2016).

\bibitem{ref5} DiCarlo, L. et al. Demonstration of two-qubit algorithms with a superconducting quantum processor. \textit{Nature} 460, 240-244 (2009).

\bibitem{ref6} Higgins, B. L., Berry, D. W., Bartlett, S. D., Wiseman, H. M. and Pryde, G. J. Entanglement-free Heisenberg-limited phase estimation. \textit{Nature} 450, 393-396 (2007).

\bibitem{ref8} Vandersypen, L. M. K. et al. Experimental realization of Shor's quantum factoring algorithm using nuclear magnetic resonance. \textit{Nature} 414, 883-887 (2001).

\bibitem{ref9} Gulde, S. et al. Implementation of the Deutsch-Jozsa algorithm on an ion-trap quantum computer. \textit{Nature} 421, 48-50 (2003).

\bibitem{ref10} IBM. \textit{A new way of thinking: The IBM quantum experience}. URL: http://www.research.ibm.com/quantum. (2017).

\bibitem{ref11} Gyongyosi, L., Imre, S. and Nguyen, H. V. A Survey on Quantum Channel Capacities, \textit{IEEE Communications Surveys and Tutorials} \textbf{99}, 1, doi: 10.1109/COMST.2017.2786748 (2018).

\bibitem{refsur} Gyongyosi, L. and Imre, S. A Survey on Quantum Computing Technology, \textit{Computer Science Review}, Elsevier, DOI: 10.1016/j. Cosrev.2018.11.002, ISSN: 1574-0137, (2018).

\bibitem{ref12} Biamonte, J. et al. Quantum Machine Learning. \textit{Nature}, 549, 195-202 (2017).

\bibitem{ref13} Lloyd, S., Mohseni, M. and Rebentrost, P. Quantum principal component analysis. \textit{Nature Physics}, 10, 631 (2014).

\bibitem{ref14} Sheng, Y. B., Zhou, L. Distributed secure quantum machine learning. \textit{Science}, 62, 1025-2019 (2017).

\bibitem{ref15} Kimble, H. J. The quantum Internet. \textit{Nature}, 453:1023-1030 (2008).

\bibitem{ref7} Farhi, E. and Neven, H. Classification with Quantum Neural Networks on Near Term Processors, \textit{arXiv:1802.06002v1} (2018).

\bibitem{ref17} Farhi, E., Goldstone, J., Gutmann, S. and Neven, H. Quantum Algorithms for Fixed Qubit Architectures. \textit{arXiv:1703.06199v1} (2017).

\bibitem{ref16} Farhi, E., Goldstone, J. and Gutmann, S. A Quantum Approximate Optimization Algorithm. \textit{arXiv:1411.4028.} (2014).

\bibitem{ref18} Farhi, E., Goldstone, J. and Gutmann, S. A Quantum Approximate Optimization Algorithm Applied to a Bounded Occurrence Constraint Problem. \textit{arXiv:1412.6062}. (2014).

\bibitem{ref19} Farhi, E. and Harrow, A. W. Quantum Supremacy through the Quantum Approximate Optimization Algorithm. \textit{arxiv:1602.07674} (2016).

\bibitem{ref20} Lloyd, S., Shapiro, J. H., Wong, F. N. C., Kumar, P., Shahriar, S. M. and Yuen, H. P. Infrastructure for the quantum Internet. \textit{ACM SIGCOMM Computer Communication Review}, 34, 9-20 (2004).

\bibitem{ref21} Van Meter, R. \textit{Quantum Networking}, John Wiley and Sons Ltd, ISBN 1118648927, 9781118648926 (2014).

\bibitem{ref22} Gyongyosi, L. and Imre, S. \textit{Advanced Quantum Communications - An Engineering Approach}. Wiley-IEEE Press (New Jersey, USA), (2012).

\bibitem{ref23} Lloyd, S. Mohseni, M. and Rebentrost, P. Quantum algorithms for supervised and unsupervised machine learning. arXiv:1307.0411 (2013).

\bibitem{p1} Pirandola, S., Laurenza, R., Ottaviani, C. and Banchi, L. Fundamental limits of repeaterless quantum communications, \textit{Nature Communications}, 15043, doi:10.1038/ncomms15043 (2017).

\bibitem{p2} Pirandola, S., Braunstein, S.L., Laurenza, R., Ottaviani, C., Cope, T.P.W., Spedalieri, G. and Banchi, L. Theory of channel simulation and bounds for private communication, \textit{Quantum Sci. Technol}. 3, 035009 (2018).

\bibitem{p3} Pirandola, S. Capacities of repeater-assisted quantum communications, \textit{arXiv:1601.00966} (2016).

\bibitem{ref24} Martins, R., Lourenco, N. and Horta, N. \textit{Analog Integrated Circuit Design Automation}, Springer, ISBN 978-3-319-34059-3, ISBN 978-3-319-34060-9 (2017).

\bibitem{ref25} Martins, R., Lourenco, N. and Horta, N. Multi-objective optimization of analog integrated circuit placement hierarchy in absolute coordinates. \textit{Expert Syst. Appl. }42(23), 9137--9151 (2015).

\bibitem{ref26} Martins, R., Povoa, R., Lourenco, N. and Horta, N. Current-flow \& current-density-aware multiobjective optimization of analog IC placement. \textit{Integr. VLSI J.} (2016).

\bibitem{com1} Deb, K., Pratap, A., Agarwal, S. and Meyarivan, T. A fast and elitist multiobjective genetic algorithm: NSGA-II, \textit{IEEE Trans. Evol. Comput.}, vol. 6, no. 2, pp. 182–197, (2002).

\bibitem{ref27} Bandyopadhyay, S., Saha, S., Maulik, U. and Deb, K. A simulated annealing-based multiobjective optimization algorithm: AMOSA, \textit{IEEE Trans. Evol. Comput.} 12(3), 269--283 (2008).

\bibitem{ref28} Suman, B. and Kumar, P. A survey of simulated annealing as a tool for single and multiobjective optimization. \textit{J. Oper. Res. Soc.} 57, 1143--1160 (2006).

\bibitem{ref29} Jiang, I., Chang, H. Y. and Chang, C. L. WiT: Optimal wiring topology for electromigration avoidance, \textit{IEEE Trans. Very Large Scale Integr. Syst. }20(4), 581--592 (2012). 

\bibitem{ref30} Rocha, F. A. E. R. M., Martins, F., Lourenco, N. C. C. and Horta, N. C. G. \textit{Electronic Design Automation of Analog ICs, Combining Gradient Models with Multi-Objective Evolutionary Algorithms}, Springer (2014).

\bibitem{ref31} Moore, G. E. Cramming more components onto integrated circuits. \textit{Electronics}. (1965).

\bibitem{refa1} Perkowski, M., Lukac, M., Kerntopf, P., Pivtoraiko, M., Folgheraiter, M., Choi, Y. W., Jung-wook, K., Lee, D., Hwangbo, W. and Kim, H. A Hierarchical Approach to Computer-Aided Design of Quantum Circuits. \textit{Electrical and Computer Engineering Faculty Publications and Presentations} 228. (2003). 

\bibitem{refa2} Bravyi, S., Browne, D., Calpin, P., Campbell, E., Gosset, D. and Howard, M. Simulation of quantum circuits by low-rank stabilizer decompositions, arXiv:1808.00128 (2018).

\bibitem{int} Munoz-Coreas, E. and Thapliyal, H. Quantum Circuit Design of A T-count Optimized Integer Multiplier, \textit{IEEE Transactions on Computers}, p 1-1, DOI: 10.1109/TC.2018.2882774 (2018). 

\bibitem{tc} Gosset, D., Kliuchnikov, V., Mosca, M. and Russo, V. An algorithm for the t-count, \textit{Quantum Information and Computation}, vol. 14, no. 15-16, pp. 1261–1276, (2014).

\bibitem{int2} Thapliyal, H., Munoz-Coreas, E., Varun, T. S. S. and Humble, T. S. Quantum Circuit Designs of Integer Division Optimizing T-count and T-depth, arXiv:1809.09732 (2018).

\bibitem{div} Jamal, L. and Babu, H. M. H. Efficient approaches to design a reversible floating point divider, in \textit{2013 IEEE International Symposium on Circuits and Systems} (ISCAS2013), pp. 3004–3007, (2013).

\bibitem{logic} Zhou, X., Leung, D. W. and Chuang, I. L. Methodology for quantum logic gate construction, \textit{Phys. Rev. A}, vol. 62, p. 052316 (2000). 

\bibitem{tel} Gottesman, D., Chuang, I. L. Quantum Teleportation is a Universal Computational Primitive, \textit{Nature} 402, 390-393 (1999).

\bibitem{depth} Amy, M., Maslov, D., Mosca, M. and Roetteler, M. A meet-in-the middle algorithm for fast synthesis of depth-optimal quantum circuits, \textit{IEEE Transactions on Computer-Aided Design of Integrated Circuits and Systems}, vol. 32, no. 6, pp. 818–830, (2013).

\bibitem{ft} Paler, A., Polian, I., Nemoto, K. and Devitt, S. J. Fault-tolerant, high level quantum circuits: form, compilation and description, \textit{Quantum Science and Technology}, vol. 2, no. 2, p. 025003, (2017). 

\bibitem{refa3} Brandao, F. G. S. L., Broughton, M., Farhi, E., Gutmann, S. and Neven, H. For Fixed Control Parameters the Quantum Approximate Optimization Algorithm's Objective Function Value Concentrates for Typical Instances, \textit{arXiv:1812.04170} (2018).

\bibitem{refa4} Zhou, L.,Wang, S.-T., Choi, S., Pichler, H. and Lukin, M. D. Quantum Approximate Optimization Algorithm: Performance, Mechanism, and Implementation on Near-Term Devices, arXiv:1812.01041 (2018).

\bibitem{refa5} Lechner, W. Quantum Approximate Optimization with Parallelizable Gates, \textit{arXiv:1802.01157v2} (2018).

\bibitem{refa6} Gavin E. Crooks, Performance of the Quantum Approximate Optimization Algorithm on the Maximum Cut Problem, \textit{arXiv:1811.08419} (2018).

\bibitem{refa7} Ho, W. W., Jonay, C. and Hsieh, T. H. Ultrafast State Preparation via the Quantum Approximate Optimization Algorithm with Long Range Interactions, \textit{arXiv:1810.04817} (2018).

\bibitem{song} Song, C et al. 10-Qubit Entanglement and Parallel Logic Operations with a Superconducting Circuit, \textit{Physical Review Letters}, vol. 119, no. 18, p. 180511 (2017).

\bibitem{refgw} Goemans, M. X. and Williamson, D. P. Improved approximation algorithms for maximum cut and satisfiability problems using semidefinite programming, \textit{J. ACM} 42, 1115 (1995).


\end{thebibliography}
\end{document}